\documentclass[11pt]{article}

\usepackage[margin=1in]{geometry}
\usepackage{amsmath,amssymb,amsthm,mathtools}
\usepackage{mathrsfs}
\usepackage{etoolbox}
\usepackage{microtype}
\usepackage{booktabs}
\usepackage{capt-of}
\usepackage{tikz}
\usetikzlibrary{arrows.meta,positioning,fit}
\usepackage{hyperref}
\usepackage{xcolor}
\hypersetup{
  colorlinks=true,
  linkcolor=blue,
  citecolor=blue,
  urlcolor=blue,
  filecolor=blue,
  linkbordercolor={0 0 1},
  citebordercolor={0 0 1},
  urlbordercolor={0 0 1},
  filebordercolor={0 0 1}
}
\AtBeginDocument{\hypersetup{pdfborder={0 0 0.7}}}
\usepackage{aliascnt}

\newtheorem{theorem}{Theorem}[section]
\newaliascnt{lemma}{theorem}
\newtheorem{lemma}[lemma]{Lemma}
\aliascntresetthe{lemma}
\newaliascnt{proposition}{theorem}
\newtheorem{proposition}[proposition]{Proposition}
\aliascntresetthe{proposition}
\newaliascnt{corollary}{theorem}
\newtheorem{corollary}[corollary]{Corollary}
\aliascntresetthe{corollary}
\theoremstyle{definition}
\newaliascnt{definition}{theorem}

\aliascntresetthe{definition}
\theoremstyle{remark}
\newaliascnt{remark}{theorem}

\aliascntresetthe{remark}

\usepackage[capitalise,nameinlink,noabbrev]{cleveref}

\makeatletter
\AtBeginDocument{%
  \@ifundefined{Hy@theorem@refstepcounter}{}{%
    \patchcmd{\cref@thmnoarg}{\refstepcounter}
      {\Hy@theorem@refstepcounter}{}{}%
  }%
}
\makeatother

\newcommand{\E}{\mathbb E}
\newcommand{\Prob}{\mathbb P}
\newcommand{\Op}[1]{\mathsf{#1}}

\newcommand{\cD}{\mathscr D}
\newcommand{\cF}{\mathcal F}

\newcommand{\cW}{\mathscr W}
\newcommand{\sX}{\mathscr X}
\newcommand{\ket}[1]{\lvert #1\rangle}
\newcommand{\bra}[1]{\langle #1\rvert}

\newcommand{\norm}[1]{\left\lVert #1\right\rVert}
\newcommand{\one}{\mathbf 1}
\newcommand{\identity}{\Op{I}}
\newcommand{\Adv}{\operatorname{Adv}^{\pm}}
\newcommand{\OR}{\operatorname{OR}}
\newcommand{\Cliq}{\operatorname{Clique}}
\newcommand{\GC}{\operatorname{GC}}
\newcommand{\Sub}{\operatorname{Sub}}

\newcommand{\core}{\operatorname{core}}

\DeclareMathOperator{\supp}{supp}

\newcommand{\extensionsection}{\section}
\newcommand{\extensionsubsection}{\subsection}

\newif\ifanonymous
\ifdefined\anonymousbuild
  \anonymoustrue
\else
  \anonymousfalse
\fi

\title{Superlinear Quantum Query Lower Bounds\\
for Subgraph Detection}
\ifanonymous
  \author{}
\else
  \author{%
    Amin Shiraz Gilani\\
    \small Joint Center for Quantum Information\\
    \small and Computer Science\\
    \small University of Maryland\\
    \small \texttt{asgilani@umd.edu}
    \and
    Xingyu Zhou\\
    \small Department of Computer Science\\
    \small The University of British Columbia\\
    \small \texttt{zxingyu@cs.ubc.ca}}
\fi
\date{}
\hypersetup{
  pdftitle={Superlinear Quantum Query Lower Bounds for Subgraph Detection},
  pdfauthor={\ifanonymous\else Amin Shiraz Gilani and Xingyu Zhou\fi}
}

\begin{document}
\ifanonymous
  \hypersetup{pageanchor=false}
  \begin{titlepage}
    \maketitle
    \thispagestyle{empty}
    \begin{abstract}
Subgraph detection asks whether an $n$-vertex graph, accessed through queries to its adjacency matrix, contains a copy of a fixed graph $H$. We prove the first unconditional superlinear lower bounds on the bounded-error quantum query complexity of this problem, answering a longstanding open question. A copy of $H$ is a certificate of constant size, so the adversary method with nonnegative weights cannot prove superlinear lower bounds.

For every fixed $r\ge 4$, detecting the clique $K_r$ requires $n^{\lambda_r-o(1)}$ queries, where $\lambda_4=19/18$, the exponents $\lambda_r$ increase strictly with $r$, and $\lambda_r\ge 2-4\sqrt{2/r}+O(1/r)$. More generally, we prove superlinear lower bounds for detecting every fixed connected graph $H$ with chromatic number $c\ge 4$. These bounds approach quadratic as $c$ grows: for sufficiently large $c$, detection requires $n^{2-O(\sqrt{\log\log c/c})-o(1)}$ queries. Chromatic number alone does not characterize the quantum query complexity of subgraph detection: we show that detecting the complete bipartite graph $K_{r,r}$ requires $n^{\beta_r-o(1)}$ queries, where $\beta_{10}=181/180$ and $\beta_r\ge 2-O(1/\sqrt{r})$.

Our main technical result is a lower bound for finding an all-ones certificate from a known family when the input bits are sampled independently. Its proof combines Zhandry's compressed oracle [CRYPTO 2019] with conditioning on a randomly planted certificate, adapting an argument of Belovs [FOCS 2026]. Our hard instances are built from graphs containing many copies of the desired subgraph with limited overlap. For cliques, we use a construction of Gowers and Janzer [CPC 2021]; for complete bipartite graphs, we use a random construction.
\end{abstract}

  \end{titlepage}
  \hypersetup{pageanchor=true}
\else
  \maketitle
  \begin{abstract}
Subgraph detection asks whether an $n$-vertex graph, accessed through queries to its adjacency matrix, contains a copy of a fixed graph $H$. We prove the first unconditional superlinear lower bounds on the bounded-error quantum query complexity of this problem, answering a longstanding open question. A copy of $H$ is a certificate of constant size, so the adversary method with nonnegative weights cannot prove superlinear lower bounds.

For every fixed $r\ge 4$, detecting the clique $K_r$ requires $n^{\lambda_r-o(1)}$ queries, where $\lambda_4=19/18$, the exponents $\lambda_r$ increase strictly with $r$, and $\lambda_r\ge 2-4\sqrt{2/r}+O(1/r)$. More generally, we prove superlinear lower bounds for detecting every fixed connected graph $H$ with chromatic number $c\ge 4$. These bounds approach quadratic as $c$ grows: for sufficiently large $c$, detection requires $n^{2-O(\sqrt{\log\log c/c})-o(1)}$ queries. Chromatic number alone does not characterize the quantum query complexity of subgraph detection: we show that detecting the complete bipartite graph $K_{r,r}$ requires $n^{\beta_r-o(1)}$ queries, where $\beta_{10}=181/180$ and $\beta_r\ge 2-O(1/\sqrt{r})$.

Our main technical result is a lower bound for finding an all-ones certificate from a known family when the input bits are sampled independently. Its proof combines Zhandry's compressed oracle [CRYPTO 2019] with conditioning on a randomly planted certificate, adapting an argument of Belovs [FOCS 2026]. Our hard instances are built from graphs containing many copies of the desired subgraph with limited overlap. For cliques, we use a construction of Gowers and Janzer [CPC 2021]; for complete bipartite graphs, we use a random construction.
\end{abstract}

\fi

\newpage
\tableofcontents

\newpage
\section{Introduction}
\label{sec:introduction}

Subgraph detection is a fundamental problem in graph algorithms: given a graph
$G$ on $n$ vertices, determine whether it contains a copy of a fixed graph
$H$, not necessarily as an induced subgraph. We study this problem in the
quantum query model, where an algorithm accesses the adjacency matrix of
$G$ through queries that ask whether two vertices are adjacent, and may
make these queries in superposition. The complexity is the number of
queries needed to decide whether $H$ is present with bounded error.

Detecting a fixed subgraph provides a natural setting for studying the
power of quantum search. A proposed copy of $H$ can be verified with only
$|E(H)|$ queries, a constant independent of $n$, but the number of possible
copies grows polynomially with $n$. Moreover, these candidates overlap:
a single edge can belong to many potential copies. Understanding how
quantum algorithms exploit these overlaps has motivated a substantial
body of work on subgraph detection~\cite{ChildsEisenberg05,
MagniezSanthaSzegedy07,ChildsKothari11,Belovs12,BelovsReichardt12,
LeeMagniezSantha12,Zhu12,LeeMagniezSantha13,JefferyKothariMagniez13,
LeGall14,CaretteEtAl20,
TeraoMori24,CornelissenGilaniPatro26}.

Triangle detection illustrates the resulting algorithmic progress.
Grover search over triples of vertices gives an $O(n^{3/2})$-query
algorithm. A sequence of improvements, culminating in an $O(n^{5/4})$
bound~\cite{MagniezSanthaSzegedy07,Belovs12,LeeMagniezSantha13,LeGall14,CaretteEtAl20},
accompanied the development of several general techniques for quantum
query algorithms, including quantum walk
frameworks~\cite{Ambainis07,MNRS11,JefferyKothariMagniez13} and learning
graph frameworks~\cite{Belovs12,LeeMagniezSantha13,CaretteEtAl20}.
More generally, learning graphs give
$O(n^{2-2/k-\alpha_H})$-query algorithms for every fixed subgraph $H$ on
$k\geq 3$ vertices, where $\alpha_H>0$ is a constant depending only on
$H$~\cite{LeeMagniezSantha12,Zhu12}. Some subgraph detection problems,
including detection of every fixed path, admit $O(n)$-query
algorithms~\cite{BelovsReichardt12}.

Despite this algorithmic progress, lower bounds for detecting fixed
subgraphs had remained at the linear scale. For cliques, an $\Omega(n)$
bound follows by embedding unstructured search over $\Theta(n^2)$
possible edges. This leaves a fundamental question, highlighted by
Childs and Kothari~\cite[Section~5]{ChildsKothari11}, and reiterated in subsequent work~\cite{TeraoMori24,CornelissenGilaniPatro26}:
\begin{align*}
    \emph{Does detecting any fixed subgraph require superlinearly many quantum queries?}
\end{align*}

A central obstacle is the certificate complexity barrier for the
adversary method with nonnegative weights~\cite{Zhang04,SpalekSzegedy06}. For a
Boolean function $f$, let $C_b(f)$ be the maximum, over inputs with value
$b$, of the minimum number of input bits that certify that value. The
nonnegative adversary bound satisfies
$\operatorname{Adv}^{+}(f)\leq\sqrt{C_0(f)C_1(f)}$.
For detecting a fixed subgraph $H$, the edges of a copy form a
$1$-certificate, so $C_1(f)\leq |E(H)|=O(1)$, while
$C_0(f)\leq\binom n2$. Consequently, this method cannot prove a lower
bound larger than $O(n)$, regardless of the subgraph. The short
certificates that make candidate copies easy to verify therefore also
limit a standard approach to proving that finding them is difficult.

The general adversary bound allows negative weights, is not subject to
this barrier, and characterizes bounded-error quantum query
complexity~\cite{LMRSS11}. Explicit constructions have proved strong
lower bounds for problems over larger alphabets, including
$k$-sum~\cite{BelovsSpalek13} and triangle-sum~\cite{BelovsRosmanis14}.
These constructions allow a solution to be planted so that any
incomplete collection of its values remains distributed as before.
Examining only part of the solution therefore does not reveal that it
was planted. For subgraph detection, however, planting a specified copy
of $H$ in a random graph forces all of its required edges to be present, 
and could make the distribution of independent edges highly correlated.
It is therefore not clear how to apply these constructions to prove
superlinear lower bounds for Boolean input problems such as subgraph 
detection.

Zhandry's compressed oracle~\cite{Zhandry19} provides a further approach
by representing a quantum computation on a random input through a
superposition of databases of recorded coordinates. Applying this
representation to certificate search also requires care. A direct
analysis that sums the amplitude moved across a certificate threshold
by each query can overcount reversible changes to the database:
amplitude may cross the threshold and later return, rather than
accumulate into successful identification of a certificate.
Capturing the difficulty of finding the bits of a common
certificate requires a more refined measure of progress, as we explain
in the proof overview.

In this work, we prove the first unconditional superlinear quantum query
lower bounds for detecting fixed subgraphs. We present the results next,
followed by an overview of the framework for lower bounds and its applications.

\subsection{Results}
\label{subsec:results}

All graphs are simple and undirected. For a fixed graph $H$, let
$\Sub_H$ denote the problem of deciding whether an $n$-vertex graph,
given by its adjacency matrix, contains a copy of $H$, and write
$\Cliq_r\coloneqq\Sub_{K_r}$. Throughout, $H$ is fixed as
$n\to\infty$; constants and the $o(1)$ terms in the exponents may
depend on the size of $H$.
For a Boolean function $f$, let $Q(f)$ denote its quantum query
complexity with error at most $1/3$. In the \emph{search} version of
$\Sub_H$, an algorithm must output a copy of $H$ with probability at
least $2/3$ whenever one exists. Since a returned copy can be checked
with $|E(H)|$ queries, every lower bound for detection also holds for
search.

\paragraph{Cliques.}
Our main result gives superlinear lower bounds for every clique on at
least four vertices. The exponents increase with the clique size and
approach two.

\begin{theorem}[Clique detection]
\label{thm:clique-containment}
There is a strictly increasing sequence $(\lambda_r)_{r\ge4}$,
with $\lambda_4=19/18$ and $\lambda_r<2$, such that, for every fixed
$r\ge4$,
\begin{equation}
  Q(\Cliq_r)\ge n^{\lambda_r-o(1)}.
  \label{eq:Lambda-bound}
\end{equation}
As $r\to\infty$, the exponents satisfy
\[
  \lambda_r\ge2-\frac{4\sqrt2}{\sqrt{r}}+O(r^{-1}).
\]
The same bounds hold for $K_r$ search.
\end{theorem}

We obtain $\lambda_r$ by combining two applications of our certificate
search lower bound (\Cref{thm:certificate-search-intro} below). The
first uses \emph{clique collision}, a generalization of graph
collision. Here the algorithm knows a fixed graph $G$ and queries
only the marks on its vertices. The problem $\GC_{K_r}$ is to decide whether
the marked vertices contain a copy of $K_r$ in $G$. The case $r=2$
is the usual graph collision problem.

\begin{theorem}[Clique collision]
\label{thm:clique-collision}
For every fixed pair of integers $r>\ell\ge1$ and all sufficiently
large $n$, there is a known $r$-partite graph $G$ on $n$ vertices
such that
\begin{equation}
 Q(\GC_{K_r})\ge n^{\frac{\ell(2r-\ell+1)}{2r(\ell+1)}-o(1)}.
 \label{eq:alpha}
\end{equation}
The same lower bound holds for search.
\end{theorem}

The exponent in \cref{eq:alpha} exceeds $1/2$ exactly when
$2\le\ell<r$. To turn a lower bound for clique collision into one
for $K_r$ detection, take $\Theta(n)$ independent instances of
$\GC_{K_{r-1}}$ on the same $(r-1)$-partite graph $G$. For each
instance, add a vertex adjacent to exactly those vertices of $G$
whose marking bits are $1$ in that instance, and keep the added
vertices pairwise nonadjacent. Every copy of $K_r$ then consists of
an added vertex and a clique marked by its instance
(\Cref{sec:packing}). This extends the reduction of Balodis and
Iraids~\cite{BalodisIraids16} from graph collision to triangle
detection. Composition with OR adds $1/2$ to the exponent. In
particular, triangle collision ($r=3$ and $\ell=2$) has lower bound
exponent $5/9 - o(1)$, giving the bound
$n^{1/2+5/9-o(1)}=n^{19/18-o(1)}$ for $K_4$ detection.

The exponents we obtain for clique collision are always less than one, so this
approach cannot reach exponent $3/2$. Our second application of the
certificate search bound works directly with edge queries. It uses
a threshold on the matching number of the database edges within a
candidate clique, as explained for large cliques in \Cref{sec:overview}, and yields
exponents tending to two
(\Cref{prop:matching-clique}). Clique collision gives the better
bound for $r\le24$, while the matching-number argument gives the
better bound for $r\ge25$. We take $\lambda_r$ to be the larger of
the two exponents. \Cref{tab:clique-exponents} lists selected values, while
\Cref{fig:clique-exponents} compares the lower bounds from both methods with the
previous lower and upper bounds. \Cref{sec:optimization} gives the optimization and the proof of
\Cref{thm:clique-containment}.

\begin{figure}[htbp]
\centering
\begin{minipage}[c]{0.31\linewidth}
\centering
\begin{tabular}{@{}ccc@{}}
\toprule
$r$ & $\lambda_r$ & Decimal\\
\midrule
$4$ & $19/18$ & $1.056$\\
$5$ & $13/12$ & $1.083$\\
$6$ & $11/10$ & $1.100$\\
$7$ & $9/8$ & $1.125$\\
$10$ & $7/6$ & $1.167$\\
$25$ & $127/100$ & $1.270$\\
$50$ & $1732/1225$ & $1.414$\\
$100$ & $1787/1155$ & $1.547$\\
$250$ & $3506/2075$ & $1.690$\\
$500$ & $176793/99800$ & $1.771$\\
\bottomrule
\end{tabular}
\end{minipage}\hfill
\begin{minipage}[c]{0.66\linewidth}
\centering
\input{sections/exponent-figure}
\end{minipage}

\begin{minipage}[t]{0.31\linewidth}
\captionof{table}{Lower bound exponents for $K_r$
detection and search.}
\label{tab:clique-exponents}
\end{minipage}\hfill
\begin{minipage}[t]{0.66\linewidth}
\caption{Query exponents for $K_r$ detection, with $r$ on a logarithmic
scale. Thick curves show $\lambda_r$, the better of our two bounds;
thin curves continue each bound where the other is better.
Shading shows the improvement over the linear lower bound. Dots mark
the values in \Cref{tab:clique-exponents}.}
\label{fig:clique-exponents}
\end{minipage}
\end{figure}

For comparison, learning graphs detect $K_r$ with
$O(n^{2-2/r-g(r)})$ queries, where
$g(r)=\Theta(r^{-3})>0$~\cite{LeeMagniezSantha12,Zhu12}. Thus both
the upper and lower bound exponents tend to two, but a polynomial
gap remains for every fixed $r$: the upper bound has exponent
$2-\Theta(1/r)$, whereas our lower bound has exponent
$2-\Theta(1/\sqrt{r})$. Our methods do not give a superlinear lower bound
for triangles (\Cref{subsec:openproblems}).

\paragraph{Chromatic number.}
Our clique bounds extend to detecting connected subgraphs of large
chromatic number. A graph of chromatic number $c$ need not contain $K_c$, so
this extension cannot rely on subgraph containment alone. Instead,
we use the \emph{homomorphism core} of $H$: a smallest subgraph of
$H$ to which $H$ admits a homomorphism. The core has the same
chromatic number as $H$. If $H$ is connected and its core has a
$K_r$ minor, reductions of Dalirrooyfard and Vassilevska
Williams~\cite{DalirrooyfardWilliams22} embed $K_r$ detection on $n$
vertices into $H$ detection on $O(n)$ vertices with constant query
overhead. The reduction also preserves search: from a valid returned
copy of $H$, we can recover a $K_r$ in the original graph using the
vertex labels and the known reduction, without additional queries
(\Cref{thm:core-minor-transfer}).

Hadwiger's conjecture asserts that every graph of chromatic number
at least $r$ has a $K_r$ minor. It is known for
$r\le6$~\cite{Dirac52,RobertsonSeymourThomas93}. For large chromatic
number, a coloring theorem of Delcourt and
Postle~\cite{DelcourtPostle24} guarantees a $K_r$ minor with
$r=\Omega(c/\log\log c)$ in every graph of chromatic number $c$.
Together, these facts give the following comparisons.

\begin{theorem}[Lower bounds from chromatic number]
\label{thm:extensions}
Let $H$ be a fixed connected graph with chromatic number $c\ge2$.
Then
\begin{equation}
 Q(\Sub_H)=\Omega\left(Q(\Cliq_r)\right)
 \label{eq:intro-chromatic-comparison}
\end{equation}
for each of the following choices of $r$:
\begin{enumerate}
\item[(i)] $r=\min\{c,6\}$, unconditionally;
\item[(ii)] some $r\ge\alpha_0c/\log\log c$, for sufficiently large
  $c$, where $\alpha_0>0$ is an absolute constant;
\item[(iii)] $r=c$, assuming Hadwiger's conjecture.
\end{enumerate}
The same comparisons hold for search.
\end{theorem}

Combining these comparisons with our clique bounds shows that $H$
detection requires $n^{19/18-o(1)}$, $n^{13/12-o(1)}$, and
$n^{11/10-o(1)}$ queries when $c=4$, $c=5$, and $c\ge6$,
respectively. For sufficiently large $c$, the lower bound is
$n^{2-\alpha\sqrt{\log\log c/c}-o(1)}$, where $\alpha>0$ is an
absolute constant, so the exponent tends to two as the chromatic
number grows. These bounds also hold for search.

Any improved lower bound for clique detection also transfers to the
corresponding subgraph problems. In particular, a superlinear lower
bound for triangle detection would imply one for every connected
non-bipartite subgraph. We prove \Cref{thm:extensions} in
\Cref{sec:chromatic}.

\paragraph{Complete bipartite graphs.}
Chromatic number at least four is sufficient for superlinear query
complexity, but it is not necessary. Every bipartite graph with an
edge has core $K_2$, so reductions through the core give no
superlinear lower bound for bipartite subgraphs. Nevertheless, we
obtain nearly quadratic lower bounds for some bipartite subgraphs by
applying the certificate search bound directly to a random graph
construction, with a threshold on the matching number.

\begin{theorem}[Complete bipartite graph detection]
\label{thm:biclique-intro}
There is a strictly increasing sequence $(\beta_r)_{r\ge10}$,
with $\beta_{10}=181/180$ and $\beta_r<2$, such that, for every fixed
$r\ge10$,
\begin{equation}
  Q(\Sub_{K_{r,r}})\ge n^{\beta_r-o(1)}.
  \label{eq:biclique-intro-bound}
\end{equation}
As $r\to\infty$, the exponents satisfy
\begin{equation}
  \beta_r\ge2-\frac{4\sqrt2}{\sqrt{r}}+O(r^{-1}).
  \label{eq:biclique-intro-asymptotic}
\end{equation}
The same bounds hold for $K_{r,r}$ search, even when the input is
bipartite and its bipartition is known to the algorithm.
\end{theorem}

\paragraph{Certificate search.}
All these results follow from a single lower bound for finding an
all-ones certificate when the input bits are independent and biased.
This certificate search bound is our main technical result. Let
$q\ge1$ be a fixed integer, and let $\cF\subseteq\binom Iq$ be a
nonempty family of certificates known to the algorithm, where $I$ is the set of coordinates. The input $\boldsymbol{x}\in\{0,1\}^I$ consists
of independent $\operatorname{Bernoulli}(p)$ bits $x_j$, with
$p\coloneqq|\cF|^{-1/q}$. A certificate $C\in\cF$ is present when
$x_j=1$ for every $j\in C$, so the expected number of present
certificates is $|\cF|p^q=1$. The algorithm must output a present
certificate.

For each certificate $C$, choose a nonempty family
$\mathcal L_C\subseteq2^C$ of subsets that lie below threshold.
We require $\mathcal L_C$ to be downward closed: every subset of a
member of $\mathcal L_C$ also belongs to $\mathcal L_C$. A set
$S\subseteq I$ is below threshold for $C$ when
$S\cap C\in\mathcal L_C$. The theorem depends on the size of these
low sets and on how many certificates cross their thresholds when
one coordinate is added.

\begin{theorem}[Certificate search; informal version of
\Cref{thm:certificate-search,cor:certificate-search}]
\label{thm:certificate-search-intro}
Let $0\le d<q$ and $i\ge0$ be integers, and let $\delta>0$.
Suppose every set in every low family has at most $d$ elements.
Suppose also that, for every coordinate $j$ and every set
$S\subseteq I\setminus\{j\}$, at most $\delta\binom{|S|}{i}$
certificates $C$ cross their threshold when $j$ is added to $S$,
meaning that $S\cap C\in\mathcal L_C$ and
$(S\cup\{j\})\cap C\notin\mathcal L_C$.
For sufficiently small $p$, every quantum algorithm that outputs a
present certificate with probability at least a constant $\sigma>0$,
averaged over the sampled input and the algorithm's measurement
outcomes, must make
\[
 \Omega\left(\left(\delta p^{q-d+i/2}\right)^{-1/(i+2)}\right)
\]
queries, where the implicit constant depends only on $q$, $d$, $i$, and
$\sigma$.
\end{theorem}

The proof uses Zhandry's compressed oracle~\cite{Zhandry19}, in its
form for independent biased bits developed by Hamoudi and
Magniez~\cite{HamoudiMagniez23}. Its key ingredient is a coherent,
non-diagonal progress measure, which we obtain by conditioning on a
randomly planted certificate following Belovs~\cite{Belovs26}.
This measure retains interference between database components when
measuring progress toward a certificate. For the same thresholds, a direct estimate using a diagonal progress
measure, which depends only on the probabilities of the database
components, gives the query lower bound
\[
 \Omega\left(\left(\delta p^{1+i/2}\right)^{-1/(i+2)}\right).
\]
Our progress measure replaces $p^{1+i/2}$ by $p^{q-d+i/2}$, improving
the bound by a polynomial factor in $1/p$ whenever $d\le q-2$.
\Cref{sec:overview} illustrates this improvement on triangle collision
and explains the role of conditioning.

\subsection{Proof overview}
\label{sec:overview}
\label{subsec:overview}

All our lower bounds derive from the certificate search bound,
\Cref{thm:certificate-search-intro}. We illustrate its proof on
triangle collision, which yields the $K_4$ bound and exhibits all
the main ideas. The central ingredient is a coherent, non-diagonal
measure of the algorithm's progress toward a certificate, obtained in the compressed
oracle framework by conditioning on a randomly planted certificate. We first
explain why this measure is needed and what it gains, then sketch
the applications, which combine the certificate search bound with
graph constructions and reductions. The discussion is informal;
\Cref{sec:compressed-oracle} gives the precise statements.

\paragraph{Triangle collision.}
Let $G$ be a tripartite graph on $\Theta(n)$ vertices containing a
family $\cF$ of $n^{2-o(1)}$ triangles, with every edge lying in at
most one triangle. Such graphs are supplied by a construction of
Gowers and Janzer~\cite{GowersJanzer21}. The algorithm knows all
edges of $G$ and queries only the vertex marks; its task is to output
a triangle whose three vertices are marked. We mark each vertex
independently with probability
$p\coloneqq|\cF|^{-1/3}=n^{-2/3+o(1)}$. The expected number of
marked triangles is then one, and a second moment bound shows that
one exists with constant probability. Thus a correct algorithm has
constant success probability $\sigma$ on this random input.

\paragraph{The compressed oracle.}
Zhandry's compressed oracle~\cite{Zhandry19}, adapted to
independent biased bits by Hamoudi and
Magniez~\cite{HamoudiMagniez23}, represents the joint state of the
algorithm and the random input as a superposition over databases
$S$, each a set of \emph{recorded} coordinates. The computation begins
with the empty database, and each query can add or remove at most one
coordinate. After $t$ queries, the expected database size satisfies
$\mathbb E|S|\le2t\sqrt p$ (\Cref{lem:database-moments}).

The standard analysis~\cite{Zhandry19,ChungEtAl21,HamoudiMagniez23}
measures progress through the probability
that the database contains a solution. It bounds how quickly this
probability can grow with each query and shows that an algorithm is
unlikely to succeed when its database does not contain its output.
For finding $K$ marked coordinates, this yields the optimal bound
of order $K/\sqrt p$ queries~\cite{HamoudiMagniez23}.

\paragraph{What should count as progress.}
For triangle collision, we must distinguish \textit{identifying} a candidate
from \textit{verifying} it. Suppose the database contains the edge $uv$ of a 
triangle $uvw$. Since $uv$ lies in no other triangle in our construction, 
the algorithm uniquely determines the third vertex $w$: the triangle
has been \emph{identified}, and there is nothing left to search for.
It has not yet been \emph{verified}, however, because $w$ is marked
only with probability $p$. A faithful progress measure should
therefore give a database containing $uv$ credit $p$, and one
containing $uvw$ credit $1$. Giving $uv$ credit $0$ ignores real
progress, while giving it credit $1$ overstates its worth by a factor
$1/p$.

\paragraph{Why the standard analysis falls short.}
The standard measure assigns each database credit $1$ if it contains
a solution and $0$ otherwise, then measures progress by the average credit. We call
such a measure \emph{diagonal}: it depends only on how likely
each database is, but not on how the databases are superposed. The
difficulty is that recording a coordinate in the compressed oracle framework is reversible: once the
database contains $uv$, a query to $w$ may record $w$, and the next
query to $w$ may erase that record, even though the value of $w$ never
changes.

A diagonal measure that gives $uvw$ more credit than $uv$ therefore
rises and falls as the algorithm repeatedly queries $w$. The direct
estimate counts each rise as progress, assigning about $\sqrt p$
in amplitude to each query as though the algorithm were searching
for $w$. Giving $uv$ the same credit as $uvw$ avoids these changes,
but then an identified triangle counts as found, overstating its
worth $p$ by a factor $1/p$. Either way, the analysis charges for
searching for $w$, although $uv$ already determines it.

\paragraph{A coherent measure from a planted certificate.}
We instead base progress on the probability that a triangle is
present, meaning that all three of its vertices are marked. Queries
do not change the input, so this probability is unaffected by the
appearance or disappearance of database records. To use it, we plant
a certificate: choose $C_\star\in\cF$ uniformly, mark its three
vertices, mark every other vertex independently with probability
$p$, and run the algorithm without revealing $C_\star$. Since
$|\cF|p^3=1$, the probability that the algorithm outputs $C_\star$
is exactly $\sigma$ (\Cref{subsec:conditioning}). Thus, although the algorithm may
output any marked triangle, finding the planted one is tied exactly
to its original success probability.

For each triangle $C$, progress is associated with the probability that $C$ is
present, counted only on the part of the state in which the database
has identified $C$, that is, contains at least two vertices of $C$.
These two conditions are different: presence concerns the input and
identification concerns the database, so a database may identify $uvw$
even when $w$ is unmarked. Following Belovs~\cite{Belovs26}, we compute
this measure by conditioning the algorithm's state on $C$ being
planted (\Cref{lem:conditioning-map}). The resulting measure is
coherent rather than diagonal: it depends on how the databases
containing $uv$ and $uvw$ are superposed. The database determines
when a triangle has been identified, while the coherent measure
assigns that identification its true worth $p$.

\paragraph{The value of an identification.}
With this measure, the analysis splits at the identification threshold
(\Cref{lem:threshold-progress}). Below the threshold, the algorithm
has recorded at most one vertex of the planted triangle and can
name the triangle only by guessing. This contributes $O(p)$ to
$\sqrt\sigma$ (\Cref{lem:low-conditioning-norm}). Above the
threshold, progress changes only when a query identifies a new
triangle, with each identification contributing about $p^2$
(\Cref{lem:boundary-moment}). To see the two factors of $p$, suppose
a query to $v$ identifies the triangle $uvw$ through a recorded vertex $u$.
The vertex $v$ must be marked, which has probability $p$; the
identified triangle is then worth $p$, the probability that $w$
is marked.

Counting identifications completes the argument. A query to a vertex
$j$ can identify a triangle only if it contains $j$ and a recorded
vertex $u$. Since the edge $uj$ lies in at most one triangle, a
single query can identify at most $|S|$ triangles. After $s$ queries,
the expected database size is at most $2s\sqrt p$, so the next query
has $O(s\sqrt p)$ possible identifications on average. Summing the
square roots of their contributions bounds the progress after $t$
queries by
\[
 O\left(\sum_{s=0}^{t-1}\sqrt{p^2(s\sqrt p)}\right)=O\left(p^{5/4}t^{3/2}\right).
\]
Constant success therefore requires
$t=\Omega(p^{-5/6})=n^{5/9-o(1)}$ queries.

\paragraph{Applications.}
The applications combine the certificate search bound with known
combinatorial constructions and reductions.

\emph{From triangle collision to $K_4$.}
Take $\lfloor n/2\rfloor$ independent triangle collision instances
on $G$. For each instance, add a vertex adjacent to exactly those
vertices of $G$ whose marking bits are $1$ in that instance, and
keep the added vertices pairwise nonadjacent (\cref{fig:or-construction}). Since $G$
is tripartite, the resulting graph contains a $K_4$ if and only if
some instance is positive. Composition with OR~\cite{LMRSS11}
adds $1/2$ to the exponent, giving $1/2+5/9=19/18$. More generally,
Gowers--Janzer graphs in which every $K_\ell$ lies in at most one
$K_r$, together with a threshold at $\ell$ vertices, give
\Cref{thm:clique-collision} and the bounds for $K_{r+1}$ in the
same way.

\emph{Large cliques.}
The route through clique collision gives detection exponents below
$3/2$, so for large $r$ we query edges directly. The certificates
are the edge sets of copies of $K_r$ in a Gowers--Janzer graph in
which every $K_{2h}$ lies in at most one such copy. A threshold on
the number of recorded edges would not identify a candidate, since
many edges can share few vertices. We instead use matching number
$h$: within a candidate clique, the endpoints of an $h$-matching
form a $K_{2h}$, which determines that candidate. The
Erd\H{o}s--Gallai theorem bounds $d$, and choosing
$h\approx\sqrt{r/2}$ gives the exponent
$2-4\sqrt2/\sqrt r+O(1/r)$ (\Cref{prop:matching-clique}).

\emph{Complete bipartite graphs.}
For $K_{r,r}$, we use the same matching threshold in a random
bipartite graph whose density makes each $K_{h,h}$ extend to about
one $K_{r,r}$. Every $h$-matching then lies in polylogarithmically
many copies. This multiplicity suffices because $\delta$ enters
the bound only through $\delta^{-1/(h+1)}$. K\"onig's theorem
bounds $d$ (\Cref{thm:biclique}).

\emph{Chromatic number.}
If $H$ is connected and its core has a $K_r$ minor, reductions of
Dalirrooyfard and Vassilevska Williams~\cite{DalirrooyfardWilliams22}
embed $K_r$ detection into $H$ detection with constant query
overhead. The required minors come from known cases of Hadwiger's
conjecture and a coloring theorem of Delcourt and
Postle~\cite{DelcourtPostle24} (\Cref{thm:extensions}).

\subsection{Related work}
\label{sec:related-work}

\paragraph{The certificate barrier.}
Let $C_b(f)$ denote the largest minimum certificate size of $f$ over
inputs with value $b$. For a fixed graph $H$, $C_1(\Sub_H)\le|E(H)|$ and
$C_0(\Sub_H)\le\binom n2$. The certificate complexity
barrier~\cite{SpalekSzegedy06} therefore gives
\[
 \operatorname{Adv}^{+}(\Sub_H)
 \le\sqrt{C_0(\Sub_H)C_1(\Sub_H)}=O(n)
\]
for the adversary bound with nonnegative weights. The general adversary
bound with negative weights characterizes quantum query
complexity~\cite{LMRSS11}, but no suitable adversary matrix for
subgraph detection was known.

\paragraph{Earlier lower bounds.}
Belovs and Rosmanis proved an almost $n^{9/7}$ lower bound for the
triangle \emph{certificate structure}~\cite{BelovsRosmanis14}. A
certificate structure specifies the possible locations of a witness
while allowing different predicates on those locations. Their bound
therefore does not apply to Boolean triangle detection itself.
Balodis and Iraids reduced graph collision to triangle
detection~\cite{BalodisIraids16}, showing that
\[
 Q(\Cliq_3)
 =\Omega\left(\sqrt n
     \max_{|V(G)|=\lfloor n/3\rfloor}
       Q(\GC_{K_2,G})\right).
\]
The known $\Omega(\sqrt n)$ lower bound for graph collision gives only a
linear bound for triangles. We extend their reduction from $K_2$
collision to $K_r$ collision and prove lower bounds with exponents
strictly greater than $1/2$ for $r\ge3$.

\paragraph{Compressed oracles.}
Zhandry's compressed oracle~\cite{Zhandry19} represents a quantum
computation on a random input as a superposition of databases that
start empty. Hamoudi and Magniez~\cite{HamoudiMagniez23} develop the
representation for product distributions and apply it to Bernoulli
search. Cojocaru et al.~\cite{CojocaruEtAl23} give a representation with
two dimensions per Bernoulli coordinate, equivalent to ours up to basis
phases. Chung, Fehr, Huang, and Liao~\cite{ChungEtAl21} bound transitions
between database properties through combinatorial conditions on a
single query, and Jeffery and Zur~\cite{JefferyZur26} show that the
compressed oracle method is captured by the multiplicative adversary
method. Gilani, Wang, Wu, and Zhou~\cite{GilaniEtAl26} use compressed
oracles to prove lower bounds for triangle and cycle search in graphs
given as edge lists, and Hao, Huang, and Liu~\cite{HaoHuangLiu26} use
expected numbers of database entries in collision bounds. Our moment
bound in \Cref{lem:database-moments} is implicit in earlier analyses of
Bernoulli compressed oracles~\cite{HamoudiMagniez23,CojocaruEtAl23}; we
include a direct operator proof.

\paragraph{Conditioning.}
Belovs~\cite{Belovs26} introduces transfer operators for changing the
input distribution, together with a framework that measures progress
toward learning a hidden type. His decomposition of the transferred
state at a knowledge threshold is the basis of our argument. In our
setting, the hidden type is a certificate chosen uniformly from a known
family, whose bits are conditioned to be one.
\Cref{lem:conditioning-map} specializes his transfer construction to
this experiment, and \Cref{lem:threshold-progress} follows his
decomposition into components below and above a threshold. We bound the
component below threshold using the orthogonality of the possible
outputs and an exact formula for the norm of the restricted conditioning
map (\Cref{lem:low-conditioning-norm}). \Cref{lem:boundary-moment}
converts a count of threshold crossings into a bound on query progress
for an arbitrary certificate family; its proof uses local query
estimates and coefficient counting as in earlier compressed oracle
analyses~\cite{ChungEtAl21,GilaniEtAl26,HaoHuangLiu26}, while retaining
the dependence on the database moments. \Cref{thm:certificate-search}
packages these lemmas into a single bound that applies to any
certificate family.

\paragraph{Graphs with restricted overlaps.}
For cliques, we use graphs from Gowers and Janzer's
generalization of the Ruzsa--Szemer\'edi construction~\cite{GowersJanzer21}.
They contain many copies of $K_r$, while each copy of a prescribed
smaller clique lies in at most one of them. For complete bipartite
graphs, we use a random construction in which each matching of the
chosen size lies in at most polylogarithmically many copies.

\subsection{Open Problems}
\label{subsec:openproblems}

\paragraph{Triangles and small bipartite subgraphs.}
Triangle detection still has no superlinear lower bound. By
\Cref{cor:chromatic-comparison}, every connected non-bipartite subgraph is
at least as hard as the triangle, so a superlinear bound for triangles
would give one for every connected non-bipartite subgraph, including all
odd cycles. A
triangle has no matching of size two, so the matching argument does not
apply. By the
reduction of Balodis and Iraids~\cite{BalodisIraids16}, a lower bound of
$\omega(\sqrt n)$ for $\GC_{K_2}$ on some fixed graph with $n$ vertices
would give a superlinear bound for triangles.

Our bound for $K_{r,r}$ becomes superlinear at $r=10$. Smaller complete
bipartite graphs, including the four-cycle $K_{2,2}$, remain open.

More generally, which fixed subgraphs can be detected with $O(n)$
queries? By \Cref{thm:extensions}, no connected subgraph of chromatic
number at least four can. Among bipartite subgraphs, both behaviors
occur: fixed paths can be detected with $O(n)$ 
queries~\cite{BelovsReichardt12}, whereas large complete bipartite 
graphs cannot. For chromatic number three, our results give no superlinear bound, and the triangle is the natural first target.

\paragraph{Improving the clique exponents.}
Our asymptotic exponent for $K_r$ is $2-O(\sqrt{1/r})$, whereas learning
graphs give the upper bound exponent $2-\Theta(1/r)$. Under our framework, closing this gap may require a different threshold or a sharper bound on how progress accumulates over
successive queries. A useful threshold must control two quantities: the
largest $|S\cap C|$ below threshold, and the number of certificates
whose threshold a single query can cross from a database $S$. The
construction for complete bipartite graphs shows that exact uniqueness
of extensions is unnecessary, since polylogarithmic multiplicity leaves
the exponent unchanged. This gives more freedom in choosing the
graph. Even for $K_4$, any improvement over the exponent $19/18$ would 
be of interest.

\paragraph{Comparison with the chromatic clique.}
Does every fixed graph $H$ of chromatic number $c\ge2$ satisfy
\[
 Q(\Sub_H)=\Omega\left(Q(\Cliq_c)\right),
\]
together with the corresponding inequality for search?
\Cref{cor:chromatic-comparison} proves both comparisons for connected
$H$ with $2\le c\le6$, and Hadwiger's conjecture would extend this to
all $c$. A different
reduction or a direct quantum argument might avoid the need for clique minors.

\paragraph{Organization.}
\Cref{sec:preliminaries} collects definitions and standard facts,
including a second moment bound for the presence of a certificate.
\Cref{sec:compressed-oracle} states and proves our general lower bound
for certificate search (\Cref{thm:certificate-search}), our main
quantum ingredient. Apart from composition with OR, the remaining
arguments are combinatorial. \Cref{sec:clique-bounds} applies the bound to
clique collision and to large cliques, composes clique collision with
OR, and optimizes the resulting exponents in \Cref{sec:optimization}.
\Cref{sec:extensions} proves the bounds from chromatic number and for
complete bipartite graphs.

\section{Preliminaries}
\label{sec:preliminaries}

This section fixes notation and collects standard facts. The compressed
oracle framework is developed in \Cref{sec:compressed-oracle}.

\subsection{Graphs and asymptotic notation}
\label{subsec:graph-preliminaries}

Write $\mathbb N\coloneqq\{1,2,\ldots\}$ and $[n]\coloneqq\{1,\ldots,n\}$.
For a finite set $V$, let $\binom Vq$ denote its family of subsets of
size $q$. All graphs are finite, simple, and undirected. An embedding of
$H$ into $G$ is an injective map $V(H)\to V(G)$ that preserves
edges; its image is a copy of $H$. Nonedges of $H$ impose no
condition, so copies need not be induced. We write $H\subseteq G$
when such a copy exists, and $G[U]$ for the subgraph induced by
$U\subseteq V(G)$. A proper coloring assigns colors to vertices so that
adjacent vertices have different colors; $\chi(G)$ is the minimum
number of colors.

A graph is $r$-partite if its vertex set can be partitioned into $r$
independent sets. A matching is a set of edges with
pairwise disjoint endpoints, and $\nu(G)$ denotes its maximum size in $G$.
For an edge set $A\subseteq E(G)$, we also write $\nu(A)$ for the
maximum size of a matching contained in $A$. A vertex cover is a set
of vertices meeting every edge.

For a graph $F$ and $1\le h\le\nu(F)$, let
\begin{equation}
 m(F,h)\coloneqq\max\{|A|:A\subseteq E(F),\ \nu(A)<h\}.
 \label{eq:general-matching-extremal}
\end{equation}
Thus $m(F,h)$ is the largest number of edges in a subgraph of $F$
with matching number less than $h$, and $m(F,h)<|E(F)|$.

In query bounds, the subgraph to be detected and the threshold parameters
are fixed independently of the number $n$ of vertices, and constants in
asymptotic notation may depend on them. We absorb constant and
polylogarithmic factors into $n^{o(1)}$.

\subsection{Quantum queries and exact reductions}
\label{subsec:query-model}

An adjacency matrix input
$\boldsymbol{x}\in\{0,1\}^{\binom{[n]}2}$ represents a graph
$G_{\boldsymbol{x}}$ on $[n]$. Define
\[
 \Sub_H(\boldsymbol{x})
 \coloneqq\one_{\{H\subseteq G_{\boldsymbol{x}}\}},
 \qquad \Cliq_r\coloneqq\Sub_{K_r}.
\]
For a Boolean function $f$, let $Q(f)$ be the minimum worst-case number
of queries of an algorithm that computes $f$ with probability at least
$2/3$ on every input. For $H$ search, an algorithm must output an
embedding of $H$ with probability at least $2/3$ on every positive
input. Its query count is a worst case
over all inputs, including negative inputs, on which the output is
unrestricted. Checking injectivity and querying the required edges
of a proposed embedding turns a search algorithm into a decision algorithm
at an additional cost of at most $|E(H)|$ queries.
For every fixed connected graph $H$ with at least one edge, the two
query complexities are within a constant factor
\cite[Lemma~23]{CornelissenGilaniPatro26}.

We also use a problem that queries vertex marks in a fixed
graph $G$. For $r\ge2$, its \emph{graph collision function for $K_r$} is
\begin{equation}
 \GC_{K_r}(\boldsymbol{x})
 \coloneqq\one_{\{K_r\subseteq G[\supp(\boldsymbol{x})]\}},
 \qquad \boldsymbol{x}\in\{0,1\}^{V(G)},
 \label{eq:clique-collision-definition}
\end{equation}
where $\supp(\boldsymbol{x})\coloneqq\{v:x_v=1\}$. Only the marking bits are
queried; all adjacencies in $G$ are known to the algorithm. We write
$\GC_{K_r,G}$ when the graph needs to be explicit. At $r=2$ this is ordinary graph
collision.

For a finite nonempty coordinate set $I$ and
$\boldsymbol{x}=(x_j)_{j\in I}\in\{0,1\}^I$, we use the controlled phase oracle
\begin{equation}
 \Op{O}_{\boldsymbol{x}}\ket{j,b}\ket w
 \coloneqq(-1)^{b x_j}\ket{j,b}\ket w,
 \qquad j\in I,\quad b\in\{0,1\}.
 \label{eq:phase-query}
\end{equation}
Here $j$ specifies the queried coordinate, $b$ labels the control qubit,
and $w$ labels the remaining workspace. Conjugating the control qubit
by a Hadamard gate converts this oracle into the usual bit oracle,
with no change in query count. We write
$\Op{Z}\coloneqq\operatorname{diag}(1,-1)$ for the Pauli-$Z$ operator.
The branch $b=0$ permits a computation with at most $t$ queries to be
padded to exactly $t$. All computation independent of the input is free.
In particular, a fixed graph or an embedding map may be hardwired,
whether or not it can be constructed efficiently.

Our reductions define each bit of a new input
$\boldsymbol{y}(\boldsymbol{x})$ to be either a known constant or a
single bit of $\boldsymbol{x}$; the same bit of $\boldsymbol{x}$ may
appear in many positions. Such a reduction preserves queries. To
simulate a query to a position of $\boldsymbol{y}$, compute its source
index reversibly. If the position holds a bit of $\boldsymbol{x}$,
query that bit with the same control. If it holds a constant, apply the
known phase directly and make the query with control $b=0$. Uncomputing
the index requires no further query. Consequently, an identity $f(\boldsymbol{x})=g(\boldsymbol{y}(\boldsymbol{x}))$ of this
form implies $Q(f)\le Q(g)$. A reduction between search problems also
specifies how to turn a valid output for the target into a valid output
for the source.

We purify randomness and defer intermediate measurements. The final
measurement may be represented by orthogonal projectors after enlarging
the workspace. For vectors, $\norm{\cdot}$ is the Hilbert space norm;
for operators it is the induced norm. For Hermitian operators,
$\Op{A}\preceq\Op{B}$ means that $\Op{B}-\Op{A}$ is positive semidefinite.
We denote the partial trace over
a register $\mathscr A$ by $\operatorname{Tr}_{\mathscr A}$. We write
$[\Op{A},\Op{B}]\coloneqq\Op{A}\Op{B}-\Op{B}\Op{A}$ and omit
identity operators on registers that are not being acted upon. The
workspace $\cW$ includes every register controlled by the algorithm.
The input register defined in \Cref{subsec:oracle-preliminaries}
represents the random input and is accessible only through the oracle.

We use the following consequence of composition for the general
adversary bound with negative weights, denoted by $\Adv$.

\begin{lemma}[Composition with OR~\cite{LMRSS11}]
\label{lem:composition}
Let $d,s\in\mathbb N$, and let $g\colon\{0,1\}^d\to\{0,1\}$ be a Boolean
function. If the $s$ copies of $g$ act on disjoint input blocks, then
\begin{equation}
 Q(\OR_s\circ g^s)=\Omega\left(\sqrt{s}Q(g)\right),
 \label{eq:composition}
\end{equation}
where $\OR_s$ is OR on $s$ bits.
\end{lemma}

\begin{proof}
The general adversary bound characterizes Boolean query complexity
with bounded error: $Q(f)=\Theta(\Adv(f))$. Its composition inequality gives
\[
 \Adv(f\circ g^s)\ge\Adv(f)\Adv(g).
\]
Take $f\coloneqq\OR_s$ and use $\Adv(\OR_s)=\Theta(\sqrt{s})$.
\end{proof}

\subsection{Certificates and a second moment bound}
\label{subsec:certificate-preliminaries}

For a known family $\cF\subseteq 2^I$ and an input
$\boldsymbol{x}\in\{0,1\}^I$, a valid certificate is a member $C\in\cF$
with $C\subseteq\supp(\boldsymbol{x})$. The certificate search problem
asks for such a member. For clique collision, $\cF$ consists of the
vertex sets of copies of $K_r$ in $G$.

For subgraph detection, we restrict the input to subgraphs of a fixed
graph $G$ on $n$ vertices by fixing every nonedge of $G$ to
zero. The coordinates are then $I\coloneqq E(G)$, and the certificates
are the edge sets of copies of $H$ in $G$. This restriction
preserves queries, so its lower bounds apply to $\Sub_H$.
We count each distinct certificate
once, even when several embeddings give the same set of coordinates.

For $p\in(0,1]$, let
$\mu_p\coloneqq\bigotimes_{j\in I}\operatorname{Bernoulli}(p)$ be the
product distribution on $\{0,1\}^I$. For $C\subseteq I$,
let $\mu_{p,C}$ be the conditional distribution obtained by fixing the
coordinates of $C$ to one; all other coordinates remain independent
$\operatorname{Bernoulli}(p)$ bits. Let
\[
 Z(\boldsymbol{x})\coloneqq|\{C\in\cF:C\subseteq\supp(\boldsymbol{x})\}|
\]
count the certificates present in $\boldsymbol{x}$. If every certificate
has size $q$ and $p=|\cF|^{-1/q}$, then $\E Z=|\cF|p^q=1$. The following
lemma shows that a certificate is then present with constant probability
whenever overlapping certificates contribute little to the second moment.
We use this form for the clique and complete bipartite bounds.

\begin{lemma}[A certificate is present]
\label{lem:second-moment}
Let $\cF\subseteq\binom Iq$ be nonempty, put $p\coloneqq|\cF|^{-1/q}$,
and let
\[
 W\coloneqq\sum_{\substack{C,C'\in\cF,\ C\ne C'\\C\cap C'\ne\varnothing}}
   p^{2q-|C\cap C'|}.
\]
Then $\Prob_{\boldsymbol{x}\sim\mu_p}(Z>0)\ge1/(2+W)$.
\end{lemma}

\begin{proof}
A pair $C,C'$ is present with probability $p^{|C\cup C'|}=p^{2q-|C\cap C'|}$,
so
\[
 \E Z^2=\sum_{C,C'\in\cF}p^{2q-|C\cap C'|}.
\]
The diagonal terms sum to $|\cF|p^q=1$. The pairs of distinct disjoint
certificates contribute at most $|\cF|^2p^{2q}=1$, and the remaining
pairs contribute $W$. Hence $\E Z^2\le2+W$, and the Cauchy--Schwarz
inequality gives
\[
 \Prob(Z>0)\ge\frac{(\E Z)^2}{\E Z^2}\ge\frac1{2+W}.\qedhere
\]
\end{proof}

Consequently, if a search algorithm outputs a certificate with
probability at least $2/3$ on every input with $Z>0$, then its success
probability under $\mu_p$ is at least $2/(3(2+W))$.

\section{A lower bound for certificate search}
\label{sec:compressed-oracle}
\label{sec:certificate-bounds}

All our lower bounds reduce to one search problem. In clique
collision, the input bits indicate which vertices of a graph known to
the algorithm are marked, and the goal is to find a copy of $K_r$ whose
vertices are all marked. In subgraph search, the input bits indicate which edges are
present, and the goal is to find the desired subgraph. In both cases,
the algorithm looks for a \emph{certificate}, a set of coordinates from
a known family whose bits are all one. This section proves a lower bound
for this problem when the input bits are independent and biased.

\paragraph{Setting.}
Throughout this section, $I$ is a finite nonempty set of coordinates,
$q\ge1$ is an integer, and $\cF\subseteq\binom Iq$ is a nonempty family
of certificates known to the algorithm. For $p\in(0,1)$, the input
$\boldsymbol{x}\in\{0,1\}^I$ is drawn from the product distribution
$\mu_p$ of independent $\operatorname{Bernoulli}(p)$ bits. An algorithm
\emph{succeeds} if it outputs some $C\in\cF$ with
$C\subseteq\supp(\boldsymbol{x})$; every other output is a failure. We
write $\sigma$ for its success probability, averaged over
$\boldsymbol{x}\sim\mu_p$.

To measure how much an algorithm has learned about a certificate $C$,
we fix a \emph{low family} $\mathcal L_C\subseteq2^C$ for each
$C\in\cF$. Each low family is nonempty and \emph{downward closed}: if
$A\in\mathcal L_C$ and $A'\subseteq A$, then $A'\in\mathcal L_C$. In
particular, $\varnothing\in\mathcal L_C$. A set $S\subseteq I$ lies
\emph{below the threshold} of $C$ when $S\cap C\in\mathcal L_C$. For
$j\in I$ and $S\subseteq I\setminus\{j\}$, the certificates whose
threshold is crossed when $j$ is added to $S$ form the family
\begin{equation}
 \cF(j,S)\coloneqq
 \{C\in\cF:S\cap C\in\mathcal L_C,
                  \ (S\cup\{j\})\cap C\notin\mathcal L_C\}.
 \label{eq:crossing-certificates}
\end{equation}
Every member of $\cF(j,S)$ contains $j$. The following theorem bounds
the success probability in terms of two parameters of the thresholds:
the size of the low sets, and the number of certificates whose threshold
a single coordinate can cross.

\begin{theorem}[Certificate search]
\label{thm:certificate-search}
Let $0\le d<q$ and $i\ge0$ be integers, and let $\delta>0$. Suppose
that every set in every low family has at most $d$ elements, and that
\begin{equation}
 |\cF(j,S)|\le\delta\binom{|S|}{i}
 \qquad\text{for all $j\in I$ and $S\subseteq I\setminus\{j\}$.}
 \label{eq:boundary-multiplicity}
\end{equation}
Then every quantum algorithm that makes at most $t$ queries has success
probability $\sigma$ satisfying
\begin{equation}
 \sqrt{\frac{\sigma}{|\cF|p^q}}
 \le\sqrt{\frac{2^q}{|\cF|p^d}}
 +\sqrt{\frac{2^{i+2}\binom{q-1}{d}\delta}{i!|\cF|p^{d-i/2}}}
   t^{(i+2)/2}.
 \label{eq:certificate-search}
\end{equation}
\end{theorem}

In our applications, $p=|\cF|^{-1/q}$. One certificate is then present
in expectation, and the left side of \cref{eq:certificate-search} equals
$\sqrt\sigma$. The first term on the right equals
$2^{q/2}|\cF|^{-(1-d/q)/2}$, which is small when $|\cF|$ is large
because $d<q$. We apply the theorem in the following form.

\begin{corollary}[Query bound for certificate search]
\label{cor:certificate-search}
In the setting of \Cref{thm:certificate-search}, suppose that
$|\cF|p^q=1$ and $\sigma\ge2^{q+2}p^{q-d}$. Then
\begin{equation}
 t\ge\left(\frac{i!\sigma}
                 {2^{i+4}\binom{q-1}{d}\delta p^{q-d+i/2}}\right)^{1/(i+2)}.
 \label{eq:certificate-query-bound}
\end{equation}
\end{corollary}

\begin{proof}
Since $|\cF|p^q=1$, the left side of \cref{eq:certificate-search} is
$\sqrt\sigma$. The second assumption bounds the first term on the right
by $\sqrt\sigma/2$. Hence
\[
 \frac{\sqrt\sigma}2
 \le\sqrt{\frac{2^{i+2}\binom{q-1}{d}\delta p^{q-d+i/2}}{i!}}
   t^{(i+2)/2},
\]
and squaring and rearranging gives \cref{eq:certificate-query-bound}.
\end{proof}

\paragraph{Outline of the proof.}
Two difficulties arise. First, quantum queries act in superposition, so
there is no classical transcript of revealed bits. The compressed oracle
gives an equivalent representation as a superposition of databases,
initially empty, in which each query changes at most one coordinate
(\Cref{subsec:oracle-preliminaries}). Second, the algorithm may output
any certificate. We therefore consider an experiment that chooses a
certificate uniformly and conditions its bits to be one. The algorithm
is not told which certificate was chosen, and we count success only when
it outputs that particular certificate. An exact averaging identity
relates this probability to $\sigma$, and a conditioning map implements
the experiment on the database register (\Cref{subsec:conditioning}).
The threshold splits each conditioned state into a low part and a high
part. The low part contributes little to success, because the
conditioning map has small norm on low databases. The high part is
initially zero and grows only through query transitions that cross the
threshold (\Cref{subsec:threshold-progress}). We bound both parts in
\Cref{subsec:threshold-norms} and combine the bounds in
\Cref{subsec:certificate-proof}.

We use the compressed oracle of Zhandry~\cite{Zhandry19} and its
Bernoulli specialization by Hamoudi and Magniez~\cite{HamoudiMagniez23}.
The conditioning argument adapts Belovs's framework for lower
bounds~\cite{Belovs26}. We give all proofs in full.

\subsection{The compressed oracle}
\label{subsec:oracle-preliminaries}

We represent the random input coherently and pass to a basis in which
the initial input state is an empty database. In this basis, each query
changes at most one database entry, and the algorithm's output
distribution is unchanged. We then bound the moments of the database
size in terms of the number of queries.

\paragraph{The input register.}
An algorithm alternates controlled phase queries with unitaries
$\Op{U}_0,\Op{U}_1,\ldots$ on a workspace $\cW$.
The unitaries are independent of the input. For $t\ge1$ and a fixed
input $\boldsymbol{x}\in\{0,1\}^I$, write
\[
 \ket{\phi_t^{\boldsymbol{x}}}
 \coloneqq\Op{U}_t\Op{O}_{\boldsymbol{x}}\Op{U}_{t-1}\cdots
   \Op{U}_1\Op{O}_{\boldsymbol{x}}\Op{U}_0\ket0,
 \qquad
 \ket{\phi_0^{\boldsymbol{x}}}\coloneqq\Op{U}_0\ket0.
\]

Recall that $\mu_p$ denotes the product distribution on $\{0,1\}^I$ whose
coordinates are independent $\operatorname{Bernoulli}(p)$ bits.
We represent this distribution coherently by introducing an input
register $\sX\coloneqq(\mathbb C^2)^{\otimes |I|}$ initialized to
\begin{equation}
 \ket{\omega}
 \coloneqq\sum_{\boldsymbol{x}}\sqrt{\mu_p(\boldsymbol{x})}\ket{\boldsymbol{x}}
 =\bigl(\sqrt{1-p}\ket0+\sqrt p\ket1\bigr)^{\otimes |I|}.
 \label{eq:Bernoulli-purification}
\end{equation}
The query operator $\Op{O}$ acts on the workspace and input register
as follows:
\begin{equation}
 \Op{O}\ket{j,b}\ket w\ket{\boldsymbol{x}}
 \coloneqq(-1)^{b x_j}\ket{j,b}\ket w\ket{\boldsymbol{x}}.
 \label{eq:purified-query}
\end{equation}
The corresponding joint states are
\[
 \ket{\phi_t}
 \coloneqq(\Op{U}_t\otimes\identity_{\sX})\Op{O}\cdots
   (\Op{U}_1\otimes\identity_{\sX})\Op{O}
   \bigl(\Op{U}_0\ket0\otimes\ket{\omega}\bigr),
 \qquad
 \ket{\phi_0}\coloneqq\Op{U}_0\ket0\otimes\ket{\omega}.
\]
For $t\ge1$, this is the joint state of the workspace and input
register after the $t$-th query.

\paragraph{The compressed basis and oracle equivalence.}
For each $j\in I$, let $\cD_j$ have orthonormal basis
$\{\ket\perp,\ket\xi\}$. The state $\ket\perp$ represents an empty
database entry. Define the unitary decompression map on one coordinate by
\begin{equation}
 \Op{S}\ket\perp\coloneqq\sqrt{1-p}\ket0+\sqrt p\ket1,
 \qquad
 \Op{S}\ket\xi\coloneqq\sqrt p\ket0-\sqrt{1-p}\ket1.
 \label{eq:Fp}
\end{equation}
The state $\ket\xi$ is orthogonal to $\ket\perp$.
Write $\cD\coloneqq\bigotimes_{j\in I}\cD_j$ for the database
register, and let $\Op{F}\coloneqq\bigotimes_{j\in I}\Op{S}$ be the full
decompression map. The spaces $\cD$ and $\sX$ describe the same oracle
register in the compressed and computational bases, respectively.
The unitary $\Op{F}:\cD\to\sX$ converts between these representations.
A \emph{database} is a set $S\subseteq I$.
Its basis vector $\ket S$ has entry $\ket\xi$ on $S$ and
$\ket\perp$ elsewhere. Thus
$\ket\varnothing=\ket\perp^{\otimes |I|}$ represents the empty database.
An entry $\ket\xi$ decompresses to a superposition of $\ket0$ and
$\ket1$, so $j\in S$ does not certify that $x_j=1$.

In this basis, define the \emph{compressed oracle} $\Op{R}$ and the
joint state $\ket{\psi_t}\in\cW\otimes\cD$ by
\begin{equation}
 \Op{R}\coloneqq(\identity_\cW\otimes\Op{F}^\dagger)
       \Op{O}(\identity_\cW\otimes\Op{F}),
 \qquad
 \ket{\psi_t}\coloneqq(\identity_\cW\otimes\Op{F}^\dagger)\ket{\phi_t}.
 \label{eq:database-equivalence}
\end{equation}
On a basis state $\ket{j,b}\ket w\ket S$, the compressed query $\Op{R}$
leaves $j,b,w$ and every database coordinate other than $j$ unchanged.
If $b=0$, it is the identity. If $b=1$, the action of $\Op{R}$ on
coordinate $j$, in the basis $\{\ket\perp,\ket\xi\}$, is
\begin{equation}
 \Op{R}_j\coloneqq\Op{S}^\dagger \Op{Z}\Op{S}
   =\begin{pmatrix}
       1-2p&2\sqrt{p(1-p)}\\
       2\sqrt{p(1-p)}&2p-1
     \end{pmatrix},
 \label{eq:Rp}
\end{equation}
where $\Op{Z}$ is the Pauli-$Z$ operator. Define the
\emph{database size operator} and its binomial moments by
\begin{equation}
 \Op{N}\ket S\coloneqq|S|\ket S,
 \qquad
 \binom{\Op{N}}{i}\ket S\coloneqq\binom{|S|}{i}\ket S
 \quad(i\in\mathbb N\cup\{0\}).
 \label{eq:N-def}
\end{equation}

The next lemma relates the compressed states $\ket{\psi_t}$ to the
computations $\ket{\phi_t^{\boldsymbol{x}}}$ on individual inputs.

\begin{lemma}[Compressed oracle equivalence]
\label{lem:compressed-equivalence}
For every $t\ge0$,
\begin{equation}
 \ket{\phi_t}
 =\sum_{\boldsymbol{x}}\sqrt{\mu_p(\boldsymbol{x})}
      \ket{\phi_t^{\boldsymbol{x}}}\ket{\boldsymbol{x}}.
 \label{eq:standard-average-state}
\end{equation}
The compressed states satisfy
\[
 \ket{\psi_0}=\Op{U}_0\ket0\otimes\ket\varnothing,
 \qquad
 \ket{\psi_{t+1}}=(\Op{U}_{t+1}\otimes\identity_\cD)\Op{R}\ket{\psi_t}
 \quad(t\ge0).
\]
Moreover,
\[
 \operatorname{Tr}_{\cD}\bigl(\ket{\psi_t}\bra{\psi_t}\bigr)
 =\operatorname{Tr}_{\sX}\bigl(\ket{\phi_t}\bra{\phi_t}\bigr)
 =\E_{\boldsymbol{x}\sim\mu_p}
    \bigl[\ket{\phi_t^{\boldsymbol{x}}}\bra{\phi_t^{\boldsymbol{x}}}\bigr].
\]
\end{lemma}

\begin{proof}
On each basis vector $\ket{\boldsymbol{x}}$ of the input register,
$\Op{O}$ applies $\Op{O}_{\boldsymbol{x}}$ to the workspace and preserves
the input label. Expanding $\ket{\omega}$ in the definition of
$\ket{\phi_t}$ gives~\cref{eq:standard-average-state} by linearity.

By~\cref{eq:Fp},
\[
 \Op{F}^\dagger\ket{\omega}
 =\bigotimes_{j\in I}\Op{S}^\dagger
      \bigl(\sqrt{1-p}\ket0+\sqrt p\ket1\bigr)
 =\ket\varnothing,
\]
so $\ket{\psi_0}=\Op{U}_0\ket0\otimes\ket\varnothing$.
For $t\ge0$, \cref{eq:database-equivalence} gives
\[
 \begin{aligned}
 \ket{\psi_{t+1}}
 &=(\identity_\cW\otimes\Op{F}^\dagger)
   (\Op{U}_{t+1}\otimes\identity_{\sX})\Op{O}\ket{\phi_t}\\
 &=(\Op{U}_{t+1}\otimes\identity_\cD)
   (\identity_\cW\otimes\Op{F}^\dagger)\Op{O}
   (\identity_\cW\otimes\Op{F})\ket{\psi_t}\\
 &=(\Op{U}_{t+1}\otimes\identity_\cD)\Op{R}\ket{\psi_t}.
 \end{aligned}
\]
Finally,
\[
 \begin{aligned}
 \operatorname{Tr}_{\cD}\bigl(\ket{\psi_t}\bra{\psi_t}\bigr)
 &=\operatorname{Tr}_{\cD}\left[
   (\identity_\cW\otimes\Op{F}^\dagger)
   \ket{\phi_t}\bra{\phi_t}
   (\identity_\cW\otimes\Op{F})\right]\\
 &=\operatorname{Tr}_{\sX}\bigl(\ket{\phi_t}\bra{\phi_t}\bigr)\\
 &=\sum_{\boldsymbol{x}}\mu_p(\boldsymbol{x})
    \ket{\phi_t^{\boldsymbol{x}}}\bra{\phi_t^{\boldsymbol{x}}}.
 \end{aligned}
\]
The second equality uses the unitarity of $\Op{F}$, and the last
follows from~\cref{eq:standard-average-state} and the orthogonality
of the input basis.
\end{proof}

\paragraph{Moments of the database size.}

The next lemma bounds the expectation and higher binomial moments
of the database size after $t$ queries.

\begin{lemma}[Moments of the database size]
\label{lem:database-moments}
After $t$ queries, the state $\ket{\psi_t}$ is supported on databases
of size at most $t$. Moreover, for every integer $i\ge1$,
\begin{equation}
  \bra{\psi_t}\binom{\Op{N}}{i}\ket{\psi_t}
  \le\binom ti\bigl(2\sqrt{p(1-p)}\bigr)^i,
  \label{eq:factorial-bound}
\end{equation}
where $\binom ti=0$ when $i>t$.
\end{lemma}

\begin{proof}
The database starts empty, and a query changes only its queried
coordinate. Workspace unitaries do not act on the database register,
proving the support assertion.

We first bound the change at a fixed coordinate $j\in I$. Let
$\Op{\Pi}_\xi\coloneqq\ket\xi\bra\xi$. \Cref{eq:Rp} gives
\begin{equation}
  \Op{R}_j^\dagger \Op{\Pi}_\xi\Op{R}_j-\Op{\Pi}_\xi
  =2\sqrt{p(1-p)}
    \begin{pmatrix}
      2\sqrt{p(1-p)}&-(1-2p)\\
      -(1-2p)&-2\sqrt{p(1-p)}
    \end{pmatrix}.
  \label{eq:coordinate-projector-increment}
\end{equation}
The $2\times2$ matrix on the right squares to the identity and has zero
trace, so its eigenvalues are $1$ and $-1$. Thus the operator on the left
has eigenvalues $\pm2\sqrt{p(1-p)}$.

Now fix $i\ge1$ and consider the query basis state $\ket{j,1}$.
On the database register, let $\Op{\Pi}_j$ act as $\Op{\Pi}_\xi$ on
coordinate $j$ and as the identity elsewhere. Then
$\Op{N}_j\coloneqq\Op{N}-\Op{\Pi}_j$ counts the entries in state
$\ket\xi$ at coordinates other than $j$.
Pascal's identity on each database basis vector gives
\[
 \binom{\Op{N}}i
 =\binom{\Op{N}_j}i
   +\Op{\Pi}_j\binom{\Op{N}_j}{i-1}.
\]
On this query block, $\Op{R}$ acts only on coordinate $j$, so it
commutes with $\Op{N}_j$. Conjugating the identity above by $\Op{R}$
and subtracting the original identity gives
\[
 \begin{aligned}
 \Op{R}^\dagger\binom{\Op{N}}i\Op{R}-\binom{\Op{N}}i
 &=\bigl(\Op{R}^\dagger\Op{\Pi}_j\Op{R}-\Op{\Pi}_j\bigr)
       \binom{\Op{N}_j}{i-1}\\
 &\preceq2\sqrt{p(1-p)}\binom{\Op{N}_j}{i-1}\\
 &\preceq2\sqrt{p(1-p)}\binom{\Op{N}}{i-1}.
 \end{aligned}
\]
The first inequality follows from \cref{eq:coordinate-projector-increment},
since $\binom{\Op{N}_j}{i-1}$ is positive semidefinite and acts on coordinates
other than $j$. The second inequality holds on each database basis vector because
$\binom{|S\setminus\{j\}|}{i-1}\le\binom{|S|}{i-1}$.
For $b=0$, the query is the identity and the difference is zero.
The bound holds for every $j$ and $b$. The basis states $\ket{j,b}$ are
mutually orthogonal, and $\Op{R}$ preserves their labels, so the same
inequality holds on the full joint space:
\begin{equation}
 \Op{R}^\dagger\binom{\Op{N}}i\Op{R}-\binom{\Op{N}}i
 \preceq2\sqrt{p(1-p)}\binom{\Op{N}}{i-1}.
 \label{eq:moment-operator-increment}
\end{equation}

Write $u_i(t)\coloneqq\bra{\psi_t}\binom{\Op{N}}i\ket{\psi_t}$.
The workspace unitary $\Op{U}_{t+1}$ commutes with the database
operators. Thus, for every $t\ge0$, \cref{eq:moment-operator-increment} gives
\begin{equation}
 u_i(t+1)
 =\bra{\psi_t}\Op{R}^\dagger\binom{\Op{N}}i\Op{R}\ket{\psi_t}
 \le u_i(t)+2\sqrt{p(1-p)}u_{i-1}(t).
 \label{eq:moment-recurrence}
\end{equation}
We have $u_0(t)=1$ and $u_i(0)=0$ for $i\ge1$.
Induction on $t$, using Pascal's identity again, yields
$u_i(t)\le\binom ti\bigl(2\sqrt{p(1-p)}\bigr)^i$ for all $i\ge0$.
This proves~\cref{eq:factorial-bound}.
\end{proof}

\subsection{Planting a certificate}
\label{subsec:conditioning}

For each $C\in\cF$, let $\sigma_C(\boldsymbol{x})$ be the probability
that the algorithm outputs $C$ on input $\boldsymbol{x}$. Each
certificate contains $q$ coordinates, so its bits are all one with
probability $p^q$. Recall that $\mu_{p,C}$ is $\mu_p$ conditioned on
this event: it fixes every coordinate of $C$ to one, and the remaining
coordinates are independent $\operatorname{Bernoulli}(p)$ bits. The
success probability is therefore
\begin{equation}
  \sigma
  \coloneqq\sum_{C\in\cF}\E_{\boldsymbol{x}\sim\mu_p}
       \bigl[\one_{\{C\subseteq\supp(\boldsymbol{x})\}}
                     \sigma_C(\boldsymbol{x})\bigr]
  =p^q\sum_{C\in\cF}\E_{\boldsymbol{x}\sim\mu_{p,C}}
                     [\sigma_C(\boldsymbol{x})].
  \label{eq:Palm-calc}
\end{equation}
Choose $C_\star$ uniformly from $\cF$, then
sample $\boldsymbol{x}\sim\mu_{p,C_\star}$. Run the original algorithm
on $\boldsymbol{x}$, without giving it $C_\star$, and let $Y$ denote its
output. On each fixed input $\boldsymbol{x}$, the probability of outputting
$C$ remains $\sigma_C(\boldsymbol{x})$. Conditioning on the chosen
certificate and then averaging over the input gives
\[
 \Prob[Y=C_\star]
 =\sum_{C\in\cF}\Prob[C_\star=C]\Prob[Y=C\mid C_\star=C]
 =\frac1{|\cF|}\sum_{C\in\cF}
   \E_{\boldsymbol{x}\sim\mu_{p,C}}[\sigma_C(\boldsymbol{x})]
 =\frac{\sigma}{|\cF|p^q}.
\]

We implement this conditioning directly on the database register.
For $j\in I$, let $\bra1_j$ act on input coordinate $j$ and as the
identity on the other input coordinates. By \cref{eq:standard-average-state},
\[
\begin{aligned}
 p^{-1/2}(\identity_\cW\otimes\bra1_j)\ket{\phi_t}
 &=\sum_{\boldsymbol{x}:x_j=1}
   \sqrt{\frac{\mu_p(\boldsymbol{x})}{p}}
   \ket{\phi_t^{\boldsymbol{x}}}\otimes
   \ket{\boldsymbol{x}_{I\setminus\{j\}}}\\
 &=\sum_{\boldsymbol{x}:x_j=1}
   \sqrt{\mu_{p,\{j\}}(\boldsymbol{x})}
   \ket{\phi_t^{\boldsymbol{x}}}\otimes
   \ket{\boldsymbol{x}_{I\setminus\{j\}}}.
\end{aligned}
\]
Here $\boldsymbol{x}_{I\setminus\{j\}}$ denotes the remaining input string.
The second line is the normalized joint state for $\mu_{p,\{j\}}$,
with coordinate $j$ removed.
Since $\Op{S}$ converts the compressed basis to the input basis,
the corresponding functional in the compressed representation is
\begin{equation}
 \Op{L}\coloneqq p^{-1/2}\bra1 \Op{S}
   =\bra\perp-\eta\bra\xi,
 \qquad \eta\coloneqq\sqrt{\frac{1-p}{p}}.
 \label{eq:lambda}
\end{equation}
For $C\subseteq I$, let $\overline C\coloneqq I\setminus C$ and
$\cD_{\overline C}\coloneqq\bigotimes_{j\in\overline C}\cD_j$.
Apply $\Op{L}$ to each coordinate in $C$ and the identity to the remaining
coordinates. The resulting map transfers the coherent state for $\mu_p$
to the one for $\mu_{p,C}$, with the fixed coordinates removed.
We denote it by $\Op{T}_C:\cD\to\cD_{\overline C}$ and call it the
\emph{conditioning map}. For every $S\subseteq I$, its action on a database
basis state is
\begin{equation}
 \Op{T}_C\ket S=(-\eta)^{|S\cap C|}\ket{S\setminus C},
 \label{eq:TC-subset}
\end{equation}
where the ket on the right is a basis state of $\cD_{\overline C}$.
The state transfer is proved in \cref{lem:conditioning-map}.

On the remaining database register, the conditional oracle
$\Op{R}_C$ acts as follows on a query basis state $\ket{j,b}$.
For $b=1$, it applies the phase $-1$ if $j\in C$ and applies
$\Op{R}_j$ if $j\notin C$.
For $b=0$, it is the identity.

\begin{lemma}[Conditioning the compressed computation]
\label{lem:conditioning-map}
For every $C\subseteq I$,
\begin{equation}
  (\identity_\cW\otimes \Op{T}_C)\Op{R}
  =\Op{R}_C(\identity_\cW\otimes \Op{T}_C).
  \label{eq:query-intertwining}
\end{equation}
For every unitary $\Op{V}$ on $\cW$ that is independent of the input,
\begin{equation}
  (\identity_\cW\otimes \Op{T}_C)(\Op{V}\otimes\identity_{\cD})
  =(\Op{V}\otimes\identity_{\cD_{\overline C}})(\identity_\cW\otimes \Op{T}_C).
  \label{eq:work-intertwining}
\end{equation}
Moreover, for every $t\ge0$,
\[
 (\identity_\cW\otimes\Op{T}_C)\ket{\psi_t}
 =\sum_{\boldsymbol{x}\in\{0,1\}^I}\sqrt{\mu_{p,C}(\boldsymbol{x})}
    \ket{\phi_t^{\boldsymbol{x}}}\otimes
    \bigotimes_{j\in\overline C}\Op{S}^\dagger\ket{x_j},
 \qquad
 \norm{(\identity_\cW\otimes\Op{T}_C)\ket{\psi_t}}^2=1.
\]
\end{lemma}

\begin{proof}
Consider a query basis state $\ket{j,b}$. If $b=0$, both query
operators are the identity. If $b=1$ and $j\in C$, \cref{eq:Rp} gives
\begin{equation}
 \Op{L}\Op{R}_j=p^{-1/2}\bra1\Op{Z}\Op{S}=-\Op{L},
 \label{eq:lambda-intertwines}
\end{equation}
since $\bra1\Op{Z}=-\bra1$. This matches the phase applied by $\Op{R}_C$.
If $b=1$ and $j\notin C$,
$\Op{T}_C$ leaves coordinate $j$ unchanged and commutes with $\Op{R}_j$.
This proves \cref{eq:query-intertwining}. Since $\Op{V}$ and $\Op{T}_C$
act on different registers, \cref{eq:work-intertwining} holds as well.

By \cref{eq:lambda}, $\Op{L}\Op{S}^\dagger=p^{-1/2}\bra1$.
Combining this identity with
\cref{eq:standard-average-state,eq:database-equivalence} and the definition
of $\mu_{p,C}$ gives the stated expansion. Distinct inputs in the support
of $\mu_{p,C}$ have distinct restrictions to $\overline C$.
Since $\Op{S}$ is unitary, their remaining database vectors are
orthonormal. Hence
\[
 \norm{(\identity_\cW\otimes\Op{T}_C)\ket{\psi_t}}^2
 =\sum_{\boldsymbol{x}\in\{0,1\}^I}\mu_{p,C}(\boldsymbol{x})=1.\qedhere
\]
\end{proof}

Although $\Op{T}_C\ket{\psi_t}$ is normalized,
$\Op{T}_C$ can have large norm on other vectors. In particular, after a
projection onto selected database components, normalization is no longer
available. The next argument controls precisely these projected vectors.

\subsection{Progress across a threshold}
\label{subsec:threshold-progress}

For each certificate $C\in\cF$, the low family splits the databases into
low and high sets. Let $\Op{\Pi}_{C,\mathrm{lo}}$ be the projector onto
the span of the database basis states $\ket S$ with
$S\cap C\in\mathcal L_C$, and put
$\Op{\Pi}_{C,\mathrm{hi}}\coloneqq\identity-\Op{\Pi}_{C,\mathrm{lo}}$.
Both projectors are diagonal in the database basis, which is the only
property used in this subsection. The next lemma relates success to the
norm of the conditioning map below threshold, and it bounds the growth of
the conditioned component above threshold during a query.

Consider a quantum algorithm, and let $\ket{\psi_t}$ be its compressed
state after $t$ queries. We measure the \emph{progress} after $t$
queries by the root mean square norm of the conditioned high components,
\begin{equation}
  \kappa_t
  \coloneqq\sqrt{\frac1{|\cF|}\sum_{C\in\cF}
       \norm{\Op{T}_C\Op{\Pi}_{C,\mathrm{hi}}\ket{\psi_t}}^2}.
  \label{eq:generic-kappa}
\end{equation}

\begin{lemma}[Progress across a threshold]
\label{lem:threshold-progress}
If the algorithm makes exactly $t$ queries and has success probability
$\sigma$, then
\begin{equation}
  \sqrt{\frac{\sigma}{|\cF|p^q}}
  \le\frac{\max_{C\in\cF}\norm{\Op{T}_C\Op{\Pi}_{C,\mathrm{lo}}}}
           {\sqrt{|\cF|}}+\kappa_t.
  \label{eq:generic-success-progress}
\end{equation}
Define the \emph{boundary crossing operator}
\begin{equation}
  \Op{B}_C\coloneqq \Op{T}_C[\Op{\Pi}_{C,\mathrm{hi}},\Op{R}].
  \label{eq:generic-boundary-def}
\end{equation}
For $t\ge0$, we have
\begin{equation}
  \kappa_{t+1}\le\kappa_t+
   \sqrt{\frac1{|\cF|}\sum_{C\in\cF}\norm{\Op{B}_C\ket{\psi_t}}^2}.
  \label{eq:generic-kappa-recurrence}
\end{equation}
\end{lemma}

\begin{proof}
\emph{Success forces progress.}
Let $\Op{\Pi}_C^{\mathrm{out}}$ be the output projector for $C$.
For each $C\in\cF$, applying this projector to the expansion in
\cref{lem:conditioning-map} gives
\[
 \norm{(\Op{\Pi}_C^{\mathrm{out}}\otimes\Op{T}_C)\ket{\psi_t}}^2
 =\E_{\boldsymbol{x}\sim\mu_{p,C}}[\sigma_C(\boldsymbol{x})].
\]
Averaging over $C\in\cF$ and using \cref{eq:Palm-calc}, we obtain
\begin{equation}
  \frac1{|\cF|}\sum_{C\in\cF}
       \norm{(\Op{\Pi}_C^{\mathrm{out}}\otimes\Op{T}_C)\ket{\psi_t}}^2
  =\frac{\sigma}{|\cF|p^q}.
  \label{eq:generic-planted-rms}
\end{equation}
Since $\Op{\Pi}_{C,\mathrm{lo}}+\Op{\Pi}_{C,\mathrm{hi}}=\identity$,
for every $C\in\cF$,
\[
 (\Op{\Pi}_C^{\mathrm{out}}\otimes\Op{T}_C)\ket{\psi_t}
 =(\Op{\Pi}_C^{\mathrm{out}}\otimes
   \Op{T}_C\Op{\Pi}_{C,\mathrm{lo}})\ket{\psi_t}
  +(\Op{\Pi}_C^{\mathrm{out}}\otimes
   \Op{T}_C\Op{\Pi}_{C,\mathrm{hi}})\ket{\psi_t}.
\]
Applying Minkowski's inequality to the averaged norms and using
\cref{eq:generic-planted-rms}, we obtain
\begin{align*}
 \sqrt{\frac{\sigma}{|\cF|p^q}}
 &=\sqrt{\frac1{|\cF|}\sum_{C\in\cF}
   \norm{(\Op{\Pi}_C^{\mathrm{out}}\otimes\Op{T}_C)\ket{\psi_t}}^2}\\
 &\le\sqrt{\frac1{|\cF|}\sum_{C\in\cF}
   \norm{(\Op{\Pi}_C^{\mathrm{out}}\otimes
     \Op{T}_C\Op{\Pi}_{C,\mathrm{lo}})\ket{\psi_t}}^2}
   +\sqrt{\frac1{|\cF|}\sum_{C\in\cF}
   \norm{(\Op{\Pi}_C^{\mathrm{out}}\otimes
     \Op{T}_C\Op{\Pi}_{C,\mathrm{hi}})\ket{\psi_t}}^2}.
\end{align*}

For the low term, the workspace and database operators act on different
registers, so
\begin{align*}
 \norm{(\Op{\Pi}_C^{\mathrm{out}}\otimes
   \Op{T}_C\Op{\Pi}_{C,\mathrm{lo}})\ket{\psi_t}}
 &=\norm{(\identity_\cW\otimes\Op{T}_C\Op{\Pi}_{C,\mathrm{lo}})
   (\Op{\Pi}_C^{\mathrm{out}}\otimes\identity_\cD)\ket{\psi_t}}\\
 &\le\norm{\Op{T}_C\Op{\Pi}_{C,\mathrm{lo}}}
   \norm{(\Op{\Pi}_C^{\mathrm{out}}\otimes\identity_\cD)\ket{\psi_t}}.
\end{align*}
The output projectors are mutually orthogonal and sum to at most the
identity. Hence
\[
 \sum_{C\in\cF}
   \norm{(\Op{\Pi}_C^{\mathrm{out}}\otimes\identity_\cD)\ket{\psi_t}}^2
 =\bra{\psi_t}\left(\sum_{C\in\cF}\Op{\Pi}_C^{\mathrm{out}}
   \otimes\identity_\cD\right)\ket{\psi_t}
 \le1.
\]
Therefore,
\begin{align*}
 \frac1{|\cF|}\sum_{C\in\cF}
   \norm{(\Op{\Pi}_C^{\mathrm{out}}\otimes
     \Op{T}_C\Op{\Pi}_{C,\mathrm{lo}})\ket{\psi_t}}^2
 &\le
   \frac{\max_{C\in\cF}\norm{\Op{T}_C\Op{\Pi}_{C,\mathrm{lo}}}^2}{|\cF|}
   \sum_{C\in\cF}
   \norm{(\Op{\Pi}_C^{\mathrm{out}}\otimes\identity_\cD)\ket{\psi_t}}^2\\
 &\le
   \frac{\max_{C\in\cF}\norm{\Op{T}_C\Op{\Pi}_{C,\mathrm{lo}}}^2}{|\cF|}.
\end{align*}

For the high term, an orthogonal projector cannot increase the norm, so
\[
 \sqrt{\frac1{|\cF|}\sum_{C\in\cF}
   \norm{(\Op{\Pi}_C^{\mathrm{out}}\otimes
     \Op{T}_C\Op{\Pi}_{C,\mathrm{hi}})\ket{\psi_t}}^2}
 \le\sqrt{\frac1{|\cF|}\sum_{C\in\cF}
   \norm{\Op{T}_C\Op{\Pi}_{C,\mathrm{hi}}\ket{\psi_t}}^2}
 =\kappa_t.
\]
Substituting these two bounds into Minkowski's inequality gives
\[
 \sqrt{\frac{\sigma}{|\cF|p^q}}
 \le\frac{\max_{C\in\cF}\norm{\Op{T}_C\Op{\Pi}_{C,\mathrm{lo}}}}
          {\sqrt{|\cF|}}+\kappa_t,
\]
proving \cref{eq:generic-success-progress}.

\emph{A query changes progress only at the boundary.}
Using \cref{eq:query-intertwining,eq:generic-boundary-def}, we obtain
\begin{equation}
  \Op{T}_C\Op{\Pi}_{C,\mathrm{hi}}\Op{R}
  =\Op{T}_C\Op{R}\Op{\Pi}_{C,\mathrm{hi}}
    +\Op{T}_C[\Op{\Pi}_{C,\mathrm{hi}},\Op{R}]
  =\Op{R}_C\Op{T}_C\Op{\Pi}_{C,\mathrm{hi}}+\Op{B}_C.
  \label{eq:generic-truncated-intertwining}
\end{equation}
For $t\ge0$, the recurrence in \cref{lem:compressed-equivalence} and
\cref{eq:work-intertwining} give
\[
 \Op{T}_C\Op{\Pi}_{C,\mathrm{hi}}\ket{\psi_{t+1}}=(\Op{U}_{t+1}\otimes\identity_{\cD_{\overline C}})\left(\Op{R}_C\Op{T}_C\Op{\Pi}_{C,\mathrm{hi}}\ket{\psi_t}+\Op{B}_C\ket{\psi_t}\right).
\]
Consequently,
\begin{align*}
 \kappa_{t+1}
 &=\sqrt{\frac1{|\cF|}\sum_{C\in\cF}
   \norm{\Op{R}_C\Op{T}_C\Op{\Pi}_{C,\mathrm{hi}}\ket{\psi_t}
         +\Op{B}_C\ket{\psi_t}}^2}\\
 &\le\sqrt{\frac1{|\cF|}\sum_{C\in\cF}
   \norm{\Op{R}_C\Op{T}_C\Op{\Pi}_{C,\mathrm{hi}}\ket{\psi_t}}^2}
   +\sqrt{\frac1{|\cF|}\sum_{C\in\cF}
   \norm{\Op{B}_C\ket{\psi_t}}^2}\\
 &=\kappa_t+\sqrt{\frac1{|\cF|}\sum_{C\in\cF}
   \norm{\Op{B}_C\ket{\psi_t}}^2}.
\end{align*}
The first and last equalities use the unitarity of $\Op{U}_{t+1}$ and
$\Op{R}_C$, respectively. The middle step follows from Minkowski's inequality.
This proves \cref{eq:generic-kappa-recurrence}.
\end{proof}

\subsection{Norms at a threshold}
\label{subsec:threshold-norms}

We now bound the two quantities in \Cref{lem:threshold-progress}: the
norm of the conditioning map below threshold, and the average squared
norm of the boundary crossing operators. The first has an exact formula.

\begin{lemma}[Conditioning norm below threshold]
\label{lem:low-conditioning-norm}
For each $C\in\cF$,
\begin{equation}
 \norm{\Op{T}_C\Op{\Pi}_{C,\mathrm{lo}}}^2
 =\sum_{A\in\mathcal L_C}\left(\frac{1-p}{p}\right)^{|A|}.
 \label{eq:threshold-low-norm}
\end{equation}
\end{lemma}

\begin{proof}
Fix $C\in\cF$, and recall that $\eta=\sqrt{(1-p)/p}$.

Every database has a unique decomposition $S=U\cup A$, where
$U\subseteq\overline C$ and $A\subseteq C$. By the definitions of the
conditioning map and the low projector,
\[
 \Op{T}_C\Op{\Pi}_{C,\mathrm{lo}}
 =\sum_{U\subseteq\overline C}\sum_{A\in\mathcal L_C}
   (-\eta)^{|A|}\ket U\bra{U\cup A}.
\]
The vectors $\ket{U\cup A}$ in this sum are orthonormal. Consequently,
\[
 (\Op{T}_C\Op{\Pi}_{C,\mathrm{lo}})
 (\Op{T}_C\Op{\Pi}_{C,\mathrm{lo}})^\dagger
 =\sum_{U\subseteq\overline C}\sum_{A\in\mathcal L_C}
   \eta^{2|A|}\ket U\bra U
 =\left(\sum_{A\in\mathcal L_C}\eta^{2|A|}\right)
   \identity_{\cD_{\overline C}}.
\]
Taking operator norms gives
\[
 \norm{\Op{T}_C\Op{\Pi}_{C,\mathrm{lo}}}^2
 =\sum_{A\in\mathcal L_C}\eta^{2|A|}
 =\sum_{A\in\mathcal L_C}\left(\frac{1-p}{p}\right)^{|A|},
\]
proving \cref{eq:threshold-low-norm}.
\end{proof}

For the boundary crossing operators, recall the crossing family
$\cF(j,S)$ from \cref{eq:crossing-certificates}. For $j\in C$, put
\[
 \partial_j\mathcal L_C
 \coloneqq\{A\subseteq C\setminus\{j\}:A\in\mathcal L_C,
                              \ A\cup\{j\}\notin\mathcal L_C\}.
\]
Thus $C\in\cF(j,S)$ exactly when $j\in C$ and
$S\cap C\in\partial_j\mathcal L_C$.

\begin{lemma}[Boundary crossing bound]
\label{lem:boundary-moment}
Under the hypotheses of \Cref{thm:certificate-search}, for every joint
state $\ket\psi\in\cW\otimes\cD$,
\begin{equation}
 \frac1{|\cF|}\sum_{C\in\cF}\norm{\Op{B}_C\ket\psi}^2
 \le\frac{4\binom{q-1}{d}\delta}{|\cF|p^d}
       \bra\psi\binom{\Op{N}}i\ket\psi.
 \label{eq:boundary-moment-bound}
\end{equation}
\end{lemma}

\begin{proof}
We first bound $\norm{\Op{B}_C\ket\psi}^2$ for a fixed $C\in\cF$.
Recall that $\eta=\sqrt{(1-p)/p}$.
Write the component of $\ket\psi$ with $b=1$ as
\begin{equation}
 \sum_{j\in I}\sum_{S\subseteq I}
       \ket{j,1}\otimes\boldsymbol{\alpha}_{j,S}\otimes\ket S,
 \label{eq:active-state}
\end{equation}
where each $\boldsymbol{\alpha}_{j,S}$ is a possibly unnormalized vector
in the remaining work registers.
By \cref{eq:generic-boundary-def}, we can compute
$\Op{B}_C\ket\psi$ by applying
$[\Op{\Pi}_{C,\mathrm{hi}},\Op{R}]$ and then the conditioning map
$\Op{T}_C$.
Using $\Op{\Pi}_{C,\mathrm{lo}}+\Op{\Pi}_{C,\mathrm{hi}}=\identity$, we have
\begin{align*}
 [\Op{\Pi}_{C,\mathrm{hi}},\Op{R}]
 &=\Op{\Pi}_{C,\mathrm{hi}}\Op{R}
     (\Op{\Pi}_{C,\mathrm{lo}}+\Op{\Pi}_{C,\mathrm{hi}})
   -(\Op{\Pi}_{C,\mathrm{lo}}+\Op{\Pi}_{C,\mathrm{hi}})
     \Op{R}\Op{\Pi}_{C,\mathrm{hi}}\\
 &=\Op{\Pi}_{C,\mathrm{hi}}\Op{R}\Op{\Pi}_{C,\mathrm{lo}}
   -\Op{\Pi}_{C,\mathrm{lo}}\Op{R}\Op{\Pi}_{C,\mathrm{hi}}.
\end{align*}
Thus only query transitions between low and high databases contribute.
When $b=0$, the query is the identity. When $j\notin C$, it leaves
$S\cap C$ unchanged. Both cases therefore give zero.

Fix $j\in C$ and $U\subseteq\overline C$, and consider the component
with query basis state $\ket{j,1}$. The query changes only coordinate
$j$, so it pairs the database vectors
\[
 \ket{U\cup A},\qquad\ket{U\cup A\cup\{j\}},
 \qquad A\subseteq C\setminus\{j\}.
\]
If both vectors are low or both are high, the commutator vanishes on
their span. Downward closure excludes the case in which the first is
high and the second is low. The remaining pairs are exactly those with
$A\in\partial_j\mathcal L_C$.

On this pair, coordinate $j$ has basis states $\{\ket\perp,\ket\xi\}$.
The high projector acts as $\Op{\Pi}_\xi=\ket\xi\bra\xi$, and the
query acts as $\Op{R}_j$. By \cref{eq:Rp}, the commutator on this pair is
\[
 [\Op{\Pi}_\xi,\Op{R}_j]
 =\Op{\Pi}_\xi\Op{R}_j-\Op{R}_j\Op{\Pi}_\xi
 =2\sqrt{p(1-p)}\begin{pmatrix}0&-1\\1&0\end{pmatrix}.
\]
Put
\[
 a\coloneqq|A|,\qquad
 \boldsymbol{u}_A\coloneqq\boldsymbol{\alpha}_{j,U\cup A},\qquad
 \boldsymbol{v}_A\coloneqq\boldsymbol{\alpha}_{j,U\cup A\cup\{j\}}.
\]
The commutator sends the contribution from this pair to
\[
 2\sqrt{p(1-p)}\ket{j,1}\otimes
 \left(\boldsymbol{u}_A\otimes\ket{U\cup A\cup\{j\}}
       -\boldsymbol{v}_A\otimes\ket{U\cup A}\right).
\]
Applying $\Op{T}_C$ removes the coordinates in $C$. The resulting
contribution is $\ket{j,1}\otimes\boldsymbol{w}_A\otimes\ket U$, where
\begin{equation}
 \begin{aligned}
 \boldsymbol{w}_A
 &\coloneqq2\sqrt{p(1-p)}
       \bigl((-\eta)^{a+1}\boldsymbol{u}_A-(-\eta)^a\boldsymbol{v}_A\bigr)\\
 &=2(-1)^{a+1}\eta^a\sqrt{1-p}
       \bigl(\sqrt{1-p}\boldsymbol{u}_A+\sqrt p\boldsymbol{v}_A\bigr).
 \end{aligned}
 \label{eq:local-boundary-identity}
\end{equation}
The two terms account for transitions from low to high and from high to
low, respectively.

For fixed $j,U$, all these contributions have the same output labels
$\ket{j,1}\ket U$, so their workspace vectors add. Different pairs
$(j,U)$ have orthogonal output labels. Consequently,
\[
 \norm{\Op{B}_C\ket\psi}^2
 =\sum_{j\in C}\sum_{U\subseteq\overline C}
   \norm{\sum_{A\in\partial_j\mathcal L_C}\boldsymbol{w}_A}^2,
\]
where each inner sum uses the vectors $\boldsymbol{w}_A$ for that fixed $j,U$.

To bound these sums, we first use Cauchy--Schwarz to obtain
\[
 \norm{\sqrt{1-p}\boldsymbol{u}_A+\sqrt p\boldsymbol{v}_A}^2
 \le\bigl((1-p)+p\bigr)
     \bigl(\norm{\boldsymbol{u}_A}^2+\norm{\boldsymbol{v}_A}^2\bigr).
\]
Applying Cauchy--Schwarz across the sets $A$ in
\cref{eq:local-boundary-identity} then gives
\begin{equation}
 \begin{aligned}
 \norm{\sum_{A\in\partial_j\mathcal L_C}\boldsymbol{w}_A}^2
 &\le4(1-p)
   \left(\sum_{A\in\partial_j\mathcal L_C}\eta^{2|A|}\right)
   \sum_{A\in\partial_j\mathcal L_C}
      \norm{\sqrt{1-p}\boldsymbol{u}_A+\sqrt p\boldsymbol{v}_A}^2\\
 &\le4\left(\sum_{A\in\partial_j\mathcal L_C}\eta^{2|A|}\right)
   \sum_{A\in\partial_j\mathcal L_C}
      \bigl(\norm{\boldsymbol{u}_A}^2+\norm{\boldsymbol{v}_A}^2\bigr).
 \end{aligned}
 \label{eq:local-boundary-bound}
\end{equation}
The second inequality uses the preceding bound and $1-p\le1$.

Every set in $\partial_j\mathcal L_C$ has size at most $d$ and lies
in the $(q-1)$-element set $C\setminus\{j\}$. For $0\le a\le d$,
\[
 \binom{q-1}{a}
 \le\binom{q-1}{a}\binom{q-1-a}{d-a}
 =\binom{q-1}{d}\binom da.
\]
Thus the binomial theorem and $1+\eta^2=p^{-1}$ give
\begin{equation}
 \sum_{A\in\partial_j\mathcal L_C}\eta^{2|A|}
 \le\sum_{a=0}^d\binom{q-1}{a}\eta^{2a}
 \le\binom{q-1}{d}(1+\eta^2)^d
 =\frac{\binom{q-1}{d}}{p^d}.
 \label{eq:boundary-binomial-sum}
\end{equation}
Substituting this into \cref{eq:local-boundary-bound}, we obtain
\[
 \norm{\sum_{A\in\partial_j\mathcal L_C}\boldsymbol{w}_A}^2
 \le\frac{4\binom{q-1}{d}}{p^d}\sum_{A\in\partial_j\mathcal L_C}
       \bigl(\norm{\boldsymbol{\alpha}_{j,U\cup A}}^2
             +\norm{\boldsymbol{\alpha}_{j,U\cup A\cup\{j\}}}^2\bigr).
\]
Summing over $j,U$ therefore gives
\[
 \norm{\Op{B}_C\ket\psi}^2
 \le\frac{4\binom{q-1}{d}}{p^d}
   \sum_{j\in C}\sum_{U\subseteq\overline C}
   \sum_{A\in\partial_j\mathcal L_C}
       \bigl(\norm{\boldsymbol{\alpha}_{j,U\cup A}}^2
             +\norm{\boldsymbol{\alpha}_{j,U\cup A\cup\{j\}}}^2\bigr).
\]

Finally, each coefficient in this sum corresponds to a unique database
$S$. Given $j,S$, the sets are necessarily
\[
 U=S\setminus C,\qquad A=(S\cap C)\setminus\{j\}.
\]
If $j\notin S$, the coefficient occurs in the first term; if $j\in S$,
it occurs in the second. In either case, the crossing condition is
$C\in\cF(j,S\setminus\{j\})$. Each coefficient is counted exactly once, so
\begin{equation}
 \norm{\Op{B}_C\ket\psi}^2
 \le\frac{4\binom{q-1}{d}}{p^d}
   \sum_{\substack{j\in I,\ S\subseteq I\\
       C\in\cF(j,S\setminus\{j\})}}
       \norm{\boldsymbol{\alpha}_{j,S}}^2.
 \label{eq:boundary-coefficient-bound}
\end{equation}

We now average over $C$. By \cref{eq:boundary-multiplicity}, the number
of certificates contributing to a fixed pair $j,S$ satisfies
\[
 |\cF(j,S\setminus\{j\})|
 \le\delta\binom{|S\setminus\{j\}|}{i}
 \le\delta\binom{|S|}{i}.
\]
Interchanging the sums therefore gives
\begin{align*}
 \frac1{|\cF|}\sum_{C\in\cF}\norm{\Op{B}_C\ket\psi}^2
 &\le\frac{4\binom{q-1}{d}}{|\cF|p^d}
   \sum_{C\in\cF}
   \sum_{\substack{j\in I,\ S\subseteq I\\
       C\in\cF(j,S\setminus\{j\})}}
       \norm{\boldsymbol{\alpha}_{j,S}}^2\\
 &\le\frac{4\binom{q-1}{d}\delta}{|\cF|p^d}
   \sum_{j\in I}\sum_{S\subseteq I}
       \binom{|S|}{i}\norm{\boldsymbol{\alpha}_{j,S}}^2\\
 &=\frac{4\binom{q-1}{d}\delta}{|\cF|p^d}
   \left(\sum_{j,S}\bra{j,1}\otimes\boldsymbol{\alpha}_{j,S}^\dagger\otimes\bra S\right)
   \binom{\Op{N}}i
   \left(\sum_{j,S}\ket{j,1}\otimes\boldsymbol{\alpha}_{j,S}\otimes\ket S\right)\\
 &\le\frac{4\binom{q-1}{d}\delta}{|\cF|p^d}
       \bra\psi\binom{\Op{N}}i\ket\psi.
\end{align*}
The equality uses \cref{eq:N-def} and the orthogonality
of the labels $(j,S)$. The final inequality includes the nonnegative
contribution from the component with $b=0$.
This proves \cref{eq:boundary-moment-bound}.
\end{proof}

\subsection{Proof of the certificate search bound}
\label{subsec:certificate-proof}

\begin{proof}[Proof of \Cref{thm:certificate-search}]
Pad the algorithm to exactly $t$ queries, and let $\Op{\Pi}_{C,\mathrm{lo}}$,
$\Op{\Pi}_{C,\mathrm{hi}}$, $\kappa_t$, and $\Op{B}_C$ be as in
\Cref{subsec:threshold-progress}. We bound the two terms on the right of
\cref{eq:generic-success-progress}.

For the low term, every $A\in\mathcal L_C$ satisfies $|A|\le d$, so
$((1-p)/p)^{|A|}\le p^{-|A|}\le p^{-d}$. Since
$|\mathcal L_C|\le2^q$, \Cref{lem:low-conditioning-norm} gives
\[
 \max_{C\in\cF}\norm{\Op{T}_C\Op{\Pi}_{C,\mathrm{lo}}}^2\le2^qp^{-d}.
\]

For the high term, since $\varnothing\in\mathcal L_C$, the empty
database lies below every threshold, so
$\Op{\Pi}_{C,\mathrm{hi}}\ket\varnothing=0$. By
\Cref{lem:compressed-equivalence},
$\ket{\psi_0}=\Op{U}_0\ket0\otimes\ket\varnothing$, and therefore
$\kappa_0=0$. For $s\ge0$, \Cref{lem:boundary-moment,lem:database-moments}
give
\[
 \frac1{|\cF|}\sum_{C\in\cF}\norm{\Op{B}_C\ket{\psi_s}}^2
 \le\frac{4\binom{q-1}{d}\delta}{|\cF|p^d}
    \binom si\bigl(2\sqrt{p(1-p)}\bigr)^i
 \le K s^i,
 \qquad
 K\coloneqq\frac{2^{i+2}\binom{q-1}{d}\delta}{i!|\cF|p^{d-i/2}},
\]
where we used $\binom si\le s^i/i!$ and $p(1-p)\le p$, with the convention
$0^0=1$. The recurrence \cref{eq:generic-kappa-recurrence} and
$\kappa_0=0$ then give
\[
 \kappa_t\le\sum_{s=0}^{t-1}\sqrt K s^{i/2}
 \le\sqrt K t^{(i+2)/2},
\]
since $s^{i/2}\le t^{i/2}$ for $0\le s<t$. Substituting both bounds into
\cref{eq:generic-success-progress} proves \cref{eq:certificate-search}.
\end{proof}

\section{Lower bounds for cliques}
\label{sec:clique-bounds}

We prove the clique lower bounds with two applications of
\Cref{cor:certificate-search}. The first queries vertex marks. We prove a
lower bound for clique collision (\Cref{thm:clique-collision}) and
compose it with OR to obtain a bound for $K_r$ detection. This argument
gives the stronger bounds for small $r$, but its exponent tends to $3/2$
as $r\to\infty$. The second argument queries edges directly and places a
threshold on the matching number of the database edges inside each
candidate clique. Its exponent tends to two. Both arguments use graphs
constructed by Gowers and Janzer, which we describe first. We then
optimize the resulting exponents in \Cref{sec:optimization} to prove
\Cref{thm:clique-containment}.

\subsection{Graphs with unique clique extensions}
\label{sec:host-graphs}

Both arguments fix a graph with many copies of $K_r$, in which
every copy of a smaller clique $K_\ell$ lies in at most one of them. This
uniqueness serves two purposes. It bounds the number of certificates
whose threshold a single query can cross, and it limits the overlaps
between certificates, which keeps a certificate present with constant
probability.

\begin{theorem}[Gowers--Janzer~{\cite[Theorem~1.2]{GowersJanzer21}}]
\label{thm:GJ}
Fix integers $r>\ell\ge1$.  There is a constant $c>0$ such that,
for every sufficiently large $n$, some graph on $n$ vertices contains at
least
\begin{equation}
  n^\ell\exp\left(-c\sqrt{\log n}\right)
  \label{eq:GJ-count}
\end{equation}
copies of $K_r$, while every copy of $K_\ell$ in the graph is contained in
at most one copy of $K_r$.
\end{theorem}

For the reduction to detection in \Cref{sec:packing}, the graph must
also be $r$-partite. A random partition provides this while preserving
a constant fraction of the copies of $K_r$.

\begin{proposition}[Partite graphs with unique clique extensions]
\label{prop:GJ-core}
Fix integers $r>\ell\ge1$. There is a constant $c'>0$ such that,
for every sufficiently large $n$, some $r$-partite graph $G$ on
$n$ vertices contains at least
\begin{equation}
  n^\ell\exp\left(-c'\sqrt{\log n}\right)
  \label{eq:core-L-lower}
\end{equation}
copies of $K_r$, with every copy of $K_\ell$ in $G$ contained in
at most one of them. In particular, $G$ has $n^{\ell-o(1)}$
copies of $K_r$ and contains no $K_{r+1}$.
\end{proposition}

\begin{proof}
Take a graph from \Cref{thm:GJ}, color each vertex independently and
uniformly with one of $r$ colors, and delete the edges whose endpoints
have the same color. A fixed $K_r$ survives precisely when its vertices
receive distinct colors. This occurs with probability $r!/r^r$, so some
coloring retains at least this fraction of the original copies.
Fix such a coloring and call the resulting graph $G$.
Since $r$ is fixed, for sufficiently large $n$ we have
$\sqrt{\log n}\ge\log(r^r/r!)$. Therefore,
\[
  \frac{r!}{r^r}
  =\exp\left(-\log\frac{r^r}{r!}\right)
  \ge\exp\left(-\sqrt{\log n}\right),
\]
and the factor $r!/r^r$ can be absorbed into the exponential
in~\cref{eq:core-L-lower}.

Deleting edges cannot create cliques. Each $K_\ell$ in $G$ is
therefore contained in at most one $K_r$, so $G$ contains at most
$\binom n\ell=O(n^\ell)$ copies of $K_r$.
Together with \cref{eq:core-L-lower}, this gives $n^{\ell-o(1)}$ copies.
Finally, every clique in the $r$-partite graph $G$ uses at most
one vertex from each part, so $G$ contains no $K_{r+1}$.
\end{proof}

Unique extensions also limit how the cliques overlap.

\begin{lemma}[Overlaps between cliques]
\label{lem:clique-overlaps}
Let $r>\ell\ge1$, let $G$ be a graph on $n$ vertices in which every copy
of $K_\ell$ lies in at most one copy of $K_r$, and let $\cF$ be the family
of vertex sets of the copies of $K_r$ in $G$. Distinct members of $\cF$
share fewer than $\ell$ vertices. For $1\le j<\ell$, the number $a_j$ of
ordered pairs of distinct members of $\cF$ that share exactly $j$
vertices satisfies
\begin{equation}
  a_j\le |\cF|\binom rj\binom{n-j}{\ell-j}
  =O\left(|\cF|n^{\ell-j}\right).
  \label{eq:pair-count}
\end{equation}
\end{lemma}

\begin{proof}
By the unique-extension property, distinct members of $\cF$ share fewer
than $\ell$ vertices. For fixed $C\in\cF$ and $A\in\binom Cj$, the same
property gives at most $\binom{n-j}{\ell-j}$ choices of $C'$ with
$C\cap C'=A$. Summing over $C$ and $A$ proves \cref{eq:pair-count}.
\end{proof}

\subsection{Clique collision}
\label{sec:core}

Recall that $\GC_{K_r}$ asks whether the marked vertices of a graph
$G$ known to the algorithm contain a copy of $K_r$; only the vertex
marks are queried. Together with \Cref{prop:GJ-core}, the following proposition
proves \Cref{thm:clique-collision}.

\begin{proposition}[Clique collision lower bound]
\label{prop:selected-core}
Fix integers $r>\ell\ge1$, and let $G$ be an $n$-vertex graph supplied
by \Cref{prop:GJ-core}. Then
\begin{equation}
  Q(\GC_{K_r})\ge n^{\frac{\ell(2r-\ell+1)}{2r(\ell+1)}-o(1)}.
  \label{eq:selected-core-lower}
\end{equation}
The same lower bound holds for search.
\end{proposition}

\begin{proof}
Let $\cF$ be the family of vertex sets of the copies of $K_r$ in $G$, so
that $|\cF|=n^{\ell-o(1)}$. We apply \Cref{cor:certificate-search} with
coordinates $I\coloneqq V(G)$, certificate size $q\coloneqq r$, and
density $p\coloneqq|\cF|^{-1/r}$.

\emph{A marked clique is likely.}
By \Cref{lem:clique-overlaps} and $|\cF|=n^{\ell-o(1)}$, the quantity $W$
in \Cref{lem:second-moment} satisfies
\[
 W=\sum_{j=1}^{\ell-1}a_jp^{2r-j}
 \le\sum_{j=1}^{\ell-1}O\left(n^{\ell-j}|\cF|^{-1+j/r}\right)
 \le\sum_{j=1}^{\ell-1}n^{-j(1-\ell/r)+o(1)}=o(1).
\]
Consider an algorithm that makes at most $t$ queries and outputs a
marked copy of $K_r$ with probability at least $2/3$ on every positive
input. By \Cref{lem:second-moment}, its success probability satisfies
$\sigma\ge2/(3(2+W))\ge1/4$ for sufficiently large $n$.

\emph{The threshold.}
For each $C\in\cF$, let
\[
 \mathcal L_C\coloneqq\{A\subseteq C:|A|<\ell\}.
\]
This family is downward closed, and its sets have at most
$d\coloneqq\ell-1<r$ elements.

\emph{Crossings.}
Fix a queried vertex $v$ and a database $S\subseteq I\setminus\{v\}$.
A clique $C$ lies in $\cF(v,S)$ exactly when $v\in C$ and
$|S\cap C|=\ell-1$. The $\ell$ vertices of $(S\cap C)\cup\{v\}$ then form
a copy of $K_\ell$, which lies in no other member of $\cF$. Thus $C$ is
determined by the set $S\cap C\in\binom S{\ell-1}$, and
\[
 |\cF(v,S)|\le\binom{|S|}{\ell-1}.
\]
Hence \cref{eq:boundary-multiplicity} holds with $i\coloneqq\ell-1$ and
$\delta\coloneqq1$.

\emph{Conclusion.}
Since $p^{r-\ell+1}=o(1)$ and $\sigma\ge1/4$,
\Cref{cor:certificate-search} applies for sufficiently large $n$. With
$q-d+i/2=r-(\ell-1)/2$ and $r,\ell$ fixed, it gives
\[
 t=\Omega\left(p^{-\frac{2r-\ell+1}{2(\ell+1)}}\right)
 =\Omega\left(|\cF|^{\frac{2r-\ell+1}{2r(\ell+1)}}\right)
 \ge n^{\frac{\ell(2r-\ell+1)}{2r(\ell+1)}-o(1)}.
\]
For fixed $r$, a standard decision-to-search reduction converts a
$t_0$-query decision algorithm into a search algorithm using
$O(t_0\log n\log\log n)+O(1)$ queries. The polylogarithmic overhead is
absorbed into the $o(1)$ term in the exponent, so the preceding search
lower bound also holds for decision.
\end{proof}

\subsection{Composition with OR}
\label{sec:packing}

We now embed an OR of independent clique collision instances into
clique detection, extending the reduction of Balodis and
Iraids~\cite{BalodisIraids16} from $K_2$ collision to $K_3$ detection.
We add one vertex for each instance, adjacent to the vertices marked in
that instance. If the graph is $r$-partite and the added vertices
are pairwise nonadjacent, each $K_{r+1}$ consists of one added vertex
and a marked $K_r$. With $\Theta(n)$ instances, composition with OR
gains a factor of $\sqrt n$.

\begin{figure}[htbp]
\centering
\begin{tikzpicture}[x=1cm,y=0.9cm,
  vertex/.style={circle,fill=black,inner sep=1.9pt},
  unused/.style={circle,fill=black!25,inner sep=1.5pt},
  every node/.style={font=\small}]
  \draw[rounded corners=2pt,black!50] (0,0) rectangle (8,2.65);
  \node[anchor=north] at (4,-0.08) {fixed tripartite graph $G$};
  \foreach \x in {1.2,4,6.8} {
    \draw[black!25] (\x-0.55,0.25) rectangle (\x+0.55,2.35);
  }
  \foreach \x/\y in {1.2/0.5,1.2/1.9,4/0.5,4/1.0,6.8/0.5,6.8/1.9}
    \node[unused] at (\x,\y) {};
  \node[vertex,label=left:$u$] (u) at (1.2,1.15) {};
  \node[vertex,label=right:$v$] (v) at (4,1.65) {};
  \node[vertex,label=right:$w$] (w) at (6.8,1.15) {};
  \draw[semithick] (u)--(v)--(w)--(u);
  \node[vertex,label=above:$z_1$] at (1.2,4.15) {};
  \node[vertex,label=above:$z_i$] (z) at (4,4.15) {};
  \node[vertex,label=above:$z_{\lfloor n/2\rfloor}$] at (6.8,4.15) {};
  \node at (2.6,4.15) {$\cdots$};
  \node at (5.4,4.15) {$\cdots$};
  \draw[densely dashed,semithick] (z)--(u) (z)--(v) (z)--(w);
  \foreach \x/\name in {1.2/1,4/2,6.8/3}
    \node[fill=white,inner sep=1pt] at (\x,2.9) {$V_{\name}$};
  \node[fill=white,inner sep=2pt] at (2.45,2.7) {$x_u^{(i)}$};
  \node[anchor=west] at (7.4,4.15) {no edges among the $z_i$};
\end{tikzpicture}
\caption{The OR construction for $K_4$. The neighbors of the added vertex
$z_i$ are precisely the vertices of $G$ marked by the $i$th input
$\boldsymbol{x}^{(i)}$. Solid edges are fixed edges of $G$;
dashed edges are queried. Only one possible witness is shown.
Every $K_4$ has exactly one added vertex and one triangle in $G$.}
\label{fig:or-construction}
\end{figure}

\begin{lemma}[From graph collision to clique detection]
\label{lem:exact-packing}
For any integers $r,n\ge2$ and every $r$-partite graph $G$ on
$\lceil n/2\rceil$ vertices,
\begin{equation}
 Q(\Cliq_{r+1})
 =\Omega\left(\sqrt n Q(\GC_{K_r})\right).
 \label{eq:collision-composition}
\end{equation}
\end{lemma}

\begin{proof}
Put $b\coloneqq\lfloor n/2\rfloor$. Take $b$ marking strings
$\boldsymbol{x}^{(1)},\ldots,\boldsymbol{x}^{(b)}\in\{0,1\}^{V(G)}$,
and write $\boldsymbol{x}$ for their concatenation. Add $b$ pairwise
nonadjacent vertices $z_1,\ldots,z_b$ to $G$. We retain all edges of $G$
and, for $i\in[b]$, connect $z_i$ and $v\in V(G)$ precisely when
$x_v^{(i)}=1$.
The resulting graph $G_{\boldsymbol{x}}$ has
$\lceil n/2\rceil+b=n$ vertices. Let $\boldsymbol{y}(\boldsymbol{x})$
be its adjacency string. Each $\boldsymbol{x}^{(i)}$ controls only the
edges incident to $z_i$. \Cref{fig:or-construction} shows the case $r=3$.

A $K_r$ marked by $\boldsymbol{x}^{(i)}$ forms a $K_{r+1}$ together
with $z_i$. Conversely, a $K_{r+1}$ in $G_{\boldsymbol{x}}$ cannot lie
entirely in $G$, which is $r$-partite, or contain two added vertices,
which are nonadjacent. It therefore consists of one $z_i$ and a $K_r$
in $G$ marked by $\boldsymbol{x}^{(i)}$. Thus
\[
 \Cliq_{r+1}(\boldsymbol{y}(\boldsymbol{x}))
 =(\OR_b\circ\GC_{K_r}^{b})(\boldsymbol{x}).
\]

Each adjacency is either fixed or equals one marking bit, so a query to
$G_{\boldsymbol{x}}$ uses at most one query to the marking strings.
Since the $b$ strings form disjoint input blocks,
\Cref{lem:composition} gives
\[
 Q(\Cliq_{r+1})
 \ge Q\left(\OR_b\circ\GC_{K_r}^{b}\right)
 =\Omega\left(\sqrt n Q(\GC_{K_r})\right).
\]
\end{proof}

Applying this reduction to \Cref{prop:selected-core} gives the following
detection bound.

\begin{corollary}[Clique detection from graph collision]
\label{cor:collision-containment}
For fixed integers $r\ge3$ and $1\le\ell\le r-2$,
\begin{equation}
 Q(\Cliq_r)\ge
 n^{\frac12+\frac{\ell(2r-\ell-1)}{2(r-1)(\ell+1)}-o(1)}.
 \label{eq:main-vertex-bound}
\end{equation}
The same lower bound holds for $K_r$ search.
\end{corollary}

\begin{proof}
Take the $(r-1)$-partite graph from \Cref{prop:GJ-core} on
$\lceil n/2\rceil$ vertices. Since this graph has $\Theta(n)$ vertices,
\Cref{prop:selected-core,lem:exact-packing} give the claimed bound.
Verifying the edges of a returned clique takes $\binom r2=O(1)$
queries, so the bound also holds for search.
\end{proof}

\subsection{Large cliques via the matching number}
\label{sec:edge-core}

After optimization over $\ell$, the exponent in
\Cref{cor:collision-containment} tends to $3/2$ as $r\to\infty$: the
collision exponent is less than one, and the outer OR contributes $1/2$.
For large cliques, we instead query edges directly and give each
candidate clique a threshold on the matching number of the database
edges inside it. Optimizing this threshold gives exponents tending to
two.

\begin{proposition}[Clique lower bound via matching number]
\label{prop:matching-clique}
Fix integers $r,h$ with $2\le h<r/2$, and let $m$ be the maximum
number of edges in a subgraph of $K_r$ with matching number less than $h$.
Then
\begin{equation}
 Q(\Cliq_r)\ge
 n^{\frac{2h}{h+1}\left(1-\frac{m-(h-1)/2}{\binom{r}{2}}\right)-o(1)}.
 \label{eq:main-matching-bound}
\end{equation}
The same lower bound holds for $K_r$ search.
\end{proposition}

\begin{proof}
Here $m=m(K_r,h)$, as defined in \cref{eq:general-matching-extremal}. Since
$h\le\lfloor r/2\rfloor=\nu(K_r)$, where $\nu(K_r)$ is the matching number of $K_r$, we have $m<\binom r2$. Take $G$ from
\Cref{prop:GJ-core} with $\ell\coloneqq2h$. Every copy of $K_{2h}$ in $G$ lies in at
most one copy of $K_r$. As in \Cref{subsec:certificate-preliminaries},
we restrict the input to subgraphs of $G$. The coordinates are
$I\coloneqq E(G)$, and the edge sets of the copies of $K_r$ in $G$ form
a certificate family $\cF$ with certificate size $q\coloneqq\binom r2$
and
\begin{equation}
 |\cF|=n^{2h-o(1)}.
 \label{eq:edge-core-count}
\end{equation}
Put $p\coloneqq|\cF|^{-1/q}$.

\emph{A clique survives.}
Two distinct copies of $K_r$ that share $j$ vertices share $\binom j2$
edges. If $j\le1$, their edge sets are disjoint, and they do not
contribute to $W$ in \Cref{lem:second-moment}. By
\Cref{lem:clique-overlaps} with $\ell=2h$, we have $j<2h$, and
the number of ordered pairs sharing $j$ vertices is
$a_j=O(|\cF|n^{2h-j})$. Using $|\cF|=n^{2h-o(1)}$ and
$\binom j2/q\le j/r$, we obtain
\[
 W=\sum_{j=2}^{2h-1}a_jp^{2q-\binom j2}
 \le\sum_{j=2}^{2h-1}O\left(n^{2h-j}|\cF|^{-1+\binom j2/q}\right)
 \le\sum_{j=2}^{2h-1}n^{-j(1-2h/r)+o(1)}=o(1).
\]
Consider an algorithm that makes at most $t$ queries and outputs a copy
of $K_r$ with probability at least $2/3$ on every positive input. By
\Cref{lem:second-moment}, its success probability satisfies
$\sigma\ge1/4$ for sufficiently large $n$.

\emph{The threshold.}
For each $C\in\cF$, let
\[
 \mathcal L_C\coloneqq\{A\subseteq C:\nu(A)<h\}.
\]
Removing edges cannot increase the matching number, so this family is
downward closed. Its sets have at most $d\coloneqq m$ edges.

\emph{Crossings.}
Fix a queried edge $e$ and a database $S\subseteq I\setminus\{e\}$. If
$C\in\cF(e,S)$, then $\nu(S\cap C)<h\le\nu((S\cap C)\cup\{e\})$, so
every $h$-matching in $(S\cap C)\cup\{e\}$ contains $e$. Removing $e$
from such a matching leaves a matching $J\subseteq S\cap C$ of size
$h-1$ whose edges avoid the endpoints of $e$. The $2h$ endpoints of
$J\cup\{e\}$ form a copy of $K_{2h}$ inside $C$, and this copy lies in no
other certificate. For fixed $e$, each $J\in\binom S{h-1}$ therefore
determines at most one crossing certificate, and
\[
 |\cF(e,S)|\le\binom{|S|}{h-1}.
\]
Hence \cref{eq:boundary-multiplicity} holds with $i\coloneqq h-1$ and
$\delta\coloneqq1$.

\emph{Conclusion.}
Since $p^{q-m}=o(1)$ and $\sigma\ge1/4$,
\Cref{cor:certificate-search} applies for sufficiently large $n$. With
$r$ and $h$ fixed, it gives
\begin{equation}
 t=\Omega\left(p^{-\frac{q-m+(h-1)/2}{h+1}}\right)
  =\Omega\left(|\cF|^{\frac1{h+1}\left(1-\frac{m-(h-1)/2}{\binom{r}{2}}\right)}\right).
 \label{eq:edge-search-lower}
\end{equation}
Substituting \cref{eq:edge-core-count} gives the search bound. Since
$K_r$ is fixed and connected, decision and search have the same query
complexity up to a constant factor, giving the decision bound as well.
\end{proof}

\subsection{Optimizing the clique exponents}
\label{sec:optimization}

We now define $\lambda_r$ and prove
\Cref{thm:clique-containment} by comparing the exponents from graph
collision and from a threshold on the matching number. For $r\ge4$,
write
\[
 \gamma_{r,\ell}\coloneqq\frac12+
   \frac{\ell(2r-\ell-1)}{2(r-1)(\ell+1)},
 \qquad 1\le\ell\le r-2,
\]
for the exponent in \Cref{cor:collision-containment}, and write
\[
 \beta_{r,h}\coloneqq\frac{2h}{h+1}
   \left(1-\frac{m(K_r,h)-(h-1)/2}{\binom r2}\right),
 \qquad 2\le h<r/2,
\]
for the exponent in \Cref{prop:matching-clique}. Recall from
\cref{eq:general-matching-extremal} that $m(K_r,h)$ is the maximum
number of edges in a subgraph of $K_r$ with matching number less than $h$.
Define
\begin{equation}
 \lambda_r\coloneqq
 \max\left(
   \{\gamma_{r,\ell}:1\le\ell\le r-2\}
   \cup
   \{\beta_{r,h}:2\le h<r/2\}
 \right).
 \label{eq:Lambda-def}
\end{equation}
The second set is empty when $r=4$. We prove
\Cref{thm:clique-containment} by showing that $\lambda_r$ grows strictly
with $r$ and by choosing a threshold on matching number that gives
the claimed asymptotic bound.

\subsubsection{Monotonicity of the bounds from clique collision}

For $r\ge4$, we have $\gamma_{r,1}=1$, whereas
\[
 \gamma_{r,2}=\frac12+\frac{2r-3}{3(r-1)}
 =1+\frac{r-3}{6(r-1)}>1.
\]
Thus $\gamma_{r,1}<\gamma_{r,2}$, so $\ell=1$ never attains the maximum
in \cref{eq:Lambda-def}.

The exponent obtained from clique collision can be written as
\begin{equation}
 \gamma_{r,\ell}
 =\frac12+\frac{\ell}{\ell+1}
   -\frac{\ell(\ell-1)}{2(r-1)(\ell+1)}.
 \label{eq:gamma-monotone}
\end{equation}
For fixed $\ell\ge2$, the exponent increases strictly with $r$:
\begin{equation}
 \gamma_{r+1,\ell}-\gamma_{r,\ell}
 =\frac{\ell(\ell-1)}{2(\ell+1)r(r-1)}>0.
 \label{eq:gamma-increment}
\end{equation}
Every admissible $\ell$ remains admissible at $r+1$.

\subsubsection{Monotonicity of the matching bound}

Fix an admissible $h\ge2$. The factor $2h/(h+1)$ in $\beta_{r,h}$ is
independent of $r$, so it suffices to show that
$(m(K_r,h)-(h-1)/2)/\binom r2$ decreases strictly with $r$.

\begin{theorem}[Erd\H{o}s--Gallai~{\cite[Theorem~4.1]{ErdosGallai59}}]
\label{thm:erdos-gallai}
Fix integers $r,h$ with $1\le h\le\lfloor r/2\rfloor$. The maximum number
of edges in a graph on $r$ vertices with matching number less than $h$ is
\begin{equation}
 m(K_r,h)=\max\left\{\binom{2h-1}{2},
                       (h-1)\left(r-\frac h2\right)\right\}.
 \label{eq:EG-branches}
\end{equation}
\end{theorem}

The first term comes from a clique on $2h-1$ vertices, and the second
from all edges incident to a fixed set of $h-1$ vertices.
For the first term, the numerator $\binom{2h-1}{2}-(h-1)/2$ is
positive and independent of $r$, while the denominator $\binom r2$
increases strictly.
When $2h<r$, the ratios for the second term satisfy
\begin{equation}
 \frac{(h-1)(2r-h-1)}{r(r-1)}
 -\frac{(h-1)(2r-h+1)}{r(r+1)}
 =\frac{2(h-1)(r-h)}
        {r(r-1)(r+1)}>0.
 \label{eq:beta-branch-decrement}
\end{equation}
Both ratios decrease strictly, and therefore so does
their maximum. This proves that $\beta_{r,h}$ increases strictly
for fixed admissible $h$.

Now choose a parameter attaining the finite maximum defining $\lambda_r$.
Whether that parameter belongs
to the collision or matching family, it remains admissible
at $r+1$ and its exponent increases strictly. Hence
$\lambda_{r+1}>\lambda_r$. At $r=4$ the matching family is empty,
so $\lambda_4=\gamma_{4,2}=19/18$.

\subsubsection{Balancing the threshold on the matching number}

Increasing $h$ makes the factor $2h/(h+1)$ closer to two, but also
allows more edges below threshold. We first determine how $h$ can grow
if the exponent is to approach two. By \cref{eq:EG-branches},
$m(K_r,h)\ge(h-1)(r-h/2)$, so for every $2\le h<r/2$,
\[
 2-\beta_{r,h}\ge\frac2{h+1}
 +\frac{2h}{h+1}\frac{(h-1)(2r-h-1)}{r(r-1)}
 \ge\frac2{h+1}+\frac{h-1}{r-1}.
\]
The last inequality uses $2h/(h+1)\ge1$ and $2r-h-1\ge r$.
Bounded $h$ or $h=\Omega(r)$ therefore keeps the exponent a constant
distance below two. An exponent approaching two requires
$h\to\infty$ and $h=o(r)$.

The difference between the two extremal sizes in \cref{eq:EG-branches} is
\begin{equation}
 (h-1)\left(r-\frac h2\right)-\binom{2h-1}{2}
 =(h-1)\left(r-\frac{5h}{2}+1\right).
 \label{eq:EG-branch-comparison}
\end{equation}
For $h=o(r)$, this is positive for sufficiently large $r$, so
\begin{equation}
 m(K_r,h)=(h-1)\left(r-\frac h2\right).
 \label{eq:d-second-branch}
\end{equation}

The difference between two and the exponent is then
\[
 2-\beta_{r,h}=\frac2{h+1}
 +\frac{2h}{h+1}\frac{(h-1)(2r-h-1)}{r(r-1)}.
\]
For $h\to\infty$ with $h=o(r)$, the two terms are asymptotic to
$2/h$ and $4h/r$, respectively. Increasing $h$ reduces $2/h$ but
increases $4h/r$, and their sum is minimized at $h=\sqrt{r/2}$.
Choose a nearest integer, so that $h=\sqrt{r/2}+O(1)$.
For sufficiently large $r$, this choice satisfies $2\le h<r/2$.
For this choice,
\[
 \frac{2h}{h+1}
 =2-\frac2h+O(r^{-1}),
 \qquad
 \frac{m(K_r,h)-(h-1)/2}{\binom r2}
 =\frac{(h-1)(2r-h-1)}{r(r-1)}
 =\frac{2h}{r}+O(r^{-1}).
\]
Substituting into the definition of $\beta_{r,h}$ gives
\begin{equation}
 \beta_{r,h}
 =\left(2-\frac2h+O(r^{-1})\right)
   \left(1-\frac{2h}{r}+O(r^{-1})\right)
 =2-\frac2h-\frac{4h}{r}+O\left(\frac1r\right).
 \label{eq:beta-expansion}
\end{equation}
Substituting $h=\sqrt{r/2}+O(1)$ into \cref{eq:beta-expansion} gives
\begin{equation}
 \beta_{r,h}
 =2-\frac{4\sqrt2}{\sqrt{r}}+O\left(\frac1r\right).
 \label{eq:beta-optimized}
\end{equation}
This proves \Cref{thm:clique-containment}.

\section{Extensions beyond cliques}
\label{sec:extensions}

This section proves our bounds for detecting subgraphs other than cliques.
\Cref{sec:chromatic} transfers clique lower bounds to every connected
subgraph whose homomorphism core has a clique minor, and derives the
bounds in
terms of chromatic number stated in \Cref{thm:extensions}.
\Cref{sec:bicliques} applies the compressed oracle framework directly
to complete bipartite graphs, whose chromatic number is two, and proves
\Cref{thm:biclique-intro}.

\begingroup
\let\extensionsection\subsection
\let\extensionsubsection\subsubsection
\extensionsection{Subgraphs of large chromatic number}
\label{sec:chromatic}

We show that detecting a connected subgraph $H$ is at least as hard as
detecting the clique $K_{\chi(H)}$ when $2\le\chi(H)\le6$. A graph of
chromatic number $c$ need not contain $K_c$, so a clique detection
instance cannot simply be planted inside a copy of $H$. Instead, we use
reductions of Dalirrooyfard and Vassilevska
Williams~\cite{DalirrooyfardWilliams22} through a clique minor of the
homomorphism core of $H$, and we check that they preserve quantum
queries. Chromatic number enters through known cases of Hadwiger's
conjecture and coloring theorems for graphs without large clique minors,
which supply the minors we need. We first recall the graph-theoretic
background.

\extensionsubsection{Cores, minors, and Hadwiger's conjecture}
\label{sec:core-background}

A \emph{homomorphism} $f\colon H\to C$ is a map $V(H)\to V(C)$ that
sends edges to edges. Unlike an embedding, it need not be injective. For
example, a proper $c$-coloring of $H$ is exactly a homomorphism from $H$
to $K_c$, so the chromatic number $\chi(H)$ is the least $c$ for which $H$ maps to $K_c$.
Among the induced subgraphs $C\subseteq H$ that admit a homomorphism
from $H$, choose one with the fewest vertices. This is the
\emph{homomorphism core} of $H$, denoted by $\core(H)$; it is unique up
to isomorphism. For example, a bipartite graph with at least one edge maps
onto one of its edges by sending every vertex in each part of the
bipartition to the endpoint in that part. Its core is therefore $K_2$,
whereas an odd cycle is its own core. We use two basic properties.

\begin{lemma}[Properties of homomorphism cores]
\label{lem:core-properties}
Let $C\coloneqq\core(H)$. Then $\chi(C)=\chi(H)$, and $C$ is connected
if $H$ is connected.
\end{lemma}

\begin{proof}
Fix a homomorphism $f\colon H\to C$. Restricting a coloring of $H$ to
$C$ gives $\chi(C)\le\chi(H)$, and pulling a coloring of $C$ back along
$f$ gives the reverse inequality. If $H$ is connected, any two
vertices of $C$ are therefore joined by the image of a path in $H$,
which is a walk in $C$. Hence $C$ is connected.
\end{proof}

A $K_r$ minor in $C$ consists of pairwise disjoint nonempty
\emph{branch sets} $B_1,\ldots,B_r\subseteq V(C)$. Each $C[B_i]$ is
connected, and $C$ has an edge between each pair of branch sets.
Contracting each branch set gives a graph containing $K_r$. For
example, contracting two edges of a $5$-cycle leaves a triangle.

Hadwiger's conjecture asserts that every graph of chromatic number at
least $r$ has a $K_r$ minor. In words, a graph that needs $r$ colors
contains $K_r$ after contracting disjoint connected sets of vertices. We
use the known cases $2\le r\le6$:
the cases $r\le3$ are elementary, the case $r=4$ is
classical~\cite{Dirac52}, and the cases $r=5,6$ follow from the
four-color theorem, with the latter implication due to Robertson,
Seymour, and Thomas~\cite{RobertsonSeymourThomas93}. For $r\ge3$,
Delcourt and Postle proved that every graph without a $K_r$ minor
has chromatic number $O(r\log\log r)$
\cite[Theorem~1.5]{DelcourtPostle24}. This bound holds regardless of
the number of vertices.

\extensionsubsection{Encoding a clique through a core minor}

\begin{theorem}[Reduction through a clique minor of the core]
\label{thm:core-minor-transfer}
Let $H$ be a fixed connected graph, and suppose that $\core(H)$ contains
a $K_r$ minor, where $r\ge2$. For inputs on $n$ vertices,
\begin{equation}
 Q(\Sub_H)=\Omega\left(Q(\Cliq_r)\right),
 \label{eq:core-minor-same-size}
\end{equation}
and the same comparison holds for search. Consequently, if $r\ge4$,
then $H$ detection requires $n^{\lambda_r-o(1)}$ queries.
\end{theorem}

\begin{proof}
Put $C\coloneqq\core(H)$. By \Cref{lem:core-properties}, $C$ is
connected, so we may enlarge the branch sets of the $K_r$ minor until
they partition $V(C)$. Given a graph $G$ on $n$ vertices, we compose
two reductions of Dalirrooyfard and Vassilevska
Williams~\cite{DalirrooyfardWilliams22}.
\begin{enumerate}
\item \emph{From $K_r$ to $C$.} Let $G_1$ be the $C$-partite graph
  of~\cite[Theorem~2.1]{DalirrooyfardWilliams22}, which has one copy of
  $V(G)$ for each vertex of $C$. Since $C$ is a core, $G_1$ contains $C$
  if and only if $G$ contains
  $K_r$~\cite[Corollary~2.2]{DalirrooyfardWilliams22}. Their statements
  use a clique minor of maximum order, but the construction and its
  proof use only the branch sets, so they apply to our $K_r$ minor.
\item \emph{From $C$ to $H$.} Color $V(G_1)$ uniformly at random with
  $|V(C)|$ colors, and let $G_2$ be the graph
  of~\cite[Theorem~2.2]{DalirrooyfardWilliams22}, which has at most
  $|V(H)||V(G_1)|$ vertices. For every coloring, if $G_2$ contains
  $H$, then $G_1$ contains $C$. If $G_1$ contains $C$, then $G_2$
  contains $H$ whenever a fixed copy of $C$ receives a prescribed
  coloring, which happens with probability at least
  $|V(C)|^{-|V(C)|}$.
\end{enumerate}

Once the coloring is fixed, every adjacency of $G_1$ is a known constant
or one adjacency of $G$, and every adjacency of $G_2$ is a known
constant or one adjacency of $G_1$. A query to $G_2$ therefore costs at
most one query to $G$. Negative inputs stay negative for every
coloring, so running an $H$ detection algorithm on $G_2$ for a constant
number of independent colorings, with constant amplification, and
accepting if any run accepts decides whether $G$ contains $K_r$ with
bounded error. For search, from a valid returned copy of $H$ in $G_2$,
we recover a copy of $C$ in $G_1$ and then a $K_r$ in $G$ using the
vertex labels and the known reduction, without additional queries.
Verifying the clique costs $\binom r2$ queries.

The graph $G_2$ has $\Theta(n)$ vertices. Constant rescaling of the
vertex count changes the query complexity of $K_r$ detection and search
by at most a constant factor, so the reduction proves
\cref{eq:core-minor-same-size}. Finally,
\Cref{thm:clique-containment} gives the exponent for $r\ge4$.
\end{proof}

\extensionsubsection{Consequences of chromatic number}

\Cref{thm:core-minor-transfer} applies to every clique minor of the
core, and the core has the same chromatic number as $H$
(\Cref{lem:core-properties}). Known cases of Hadwiger's conjecture and
coloring bounds for graphs without large clique minors therefore turn
chromatic number into query lower bounds. The following corollary restates
\Cref{thm:extensions} together with its quantitative consequences.

\begin{corollary}[Lower bounds from chromatic number]
\label{cor:chromatic-comparison}
Let $H$ be a fixed connected graph with chromatic number $c\ge2$. Then
\begin{equation}
 Q(\Sub_H)=\Omega\left(Q(\Cliq_r)\right)
 \label{eq:chromatic-comparison}
\end{equation}
for each of the following choices of $r$:
\begin{enumerate}
\item[(i)] $r=\min\{c,6\}$;
\item[(ii)] some $r\ge\alpha_0c/\log\log c$, if $c$ is sufficiently
  large, where $\alpha_0>0$ is an absolute constant;
\item[(iii)] $r=c$, if Hadwiger's conjecture holds.
\end{enumerate}
Consequently, $H$ detection requires $n^{19/18-o(1)}$, $n^{13/12-o(1)}$,
and $n^{11/10-o(1)}$ queries if $c=4$, $c=5$, and $c\ge6$,
respectively, and $n^{2-\alpha\sqrt{\log\log c/c}-o(1)}$ queries if $c$
is sufficiently large, where $\alpha>0$ is an absolute constant. All
these comparisons and bounds also hold for search.
\end{corollary}

\begin{proof}
By \Cref{lem:core-properties}, $\core(H)$ has chromatic number $c$, so
by \Cref{thm:core-minor-transfer} it suffices to find a $K_r$ minor in
$\core(H)$ for the claimed values of $r$. For~(i), the known cases of
Hadwiger's conjecture give a $K_c$ minor if $c\le6$; if $c>6$, the core
is not $5$-colorable and therefore has a $K_6$
minor~\cite{RobertsonSeymourThomas93}. For~(iii), Hadwiger's conjecture
gives a $K_c$ minor. For~(ii), let $r$ be the largest order of a clique
minor in $\core(H)$. The coloring bound of Delcourt and
Postle~\cite[Theorem~1.5]{DelcourtPostle24} gives
$c=O((r+1)\log\log(r+1))$, and hence $r=\Omega(c/\log\log c)$.

For the explicit bounds, \Cref{cor:collision-containment} with $\ell=2$
gives
\[
 Q(\Cliq_r)\ge n^{\frac12+\frac{2r-3}{3(r-1)}-o(1)},
\]
whose exponent equals $19/18$, $13/12$, and $11/10$ for $r=4$, $5$, and
$6$; apply~(i). For sufficiently large $c$, apply~(ii). The resulting
$r$ tends to infinity with $c$, so
\Cref{thm:clique-containment} gives
\[
 \lambda_r\ge2-O(r^{-1/2})
 \ge2-O\left(\sqrt{\frac{\log\log c}{c}}\right).
\]
\Cref{thm:core-minor-transfer,cor:collision-containment,thm:clique-containment}
also hold for search, so the same argument proves the search versions.
\end{proof}

\extensionsection{Detecting complete bipartite graphs}
\label{sec:bicliques}

Every bipartite graph with an edge has core $K_2$, so
\Cref{sec:chromatic} gives no superlinear bound for detecting bipartite subgraphs.
For $K_{r,r}$, we instead apply \Cref{cor:certificate-search} directly,
with the same threshold on the matching number as in
\Cref{sec:edge-core}. We construct a graph with many copies of $K_{r,r}$
in which each matching of size $h$ lies in few copies. Retaining the edges of this graph independently at
a suitable density also leaves a copy of $K_{r,r}$ with constant
probability.

\begin{theorem}[Lower bound for $K_{r,r}$]
\label{thm:biclique}
For fixed integers $2\le h<r$, $K_{r,r}$ detection and search on $n$
vertices require
\begin{equation}
 n^{\frac{2hr}{(h+1)(r+h)}
      \left(1-\frac{h-1}{r}+\frac{h-1}{2r^2}\right)-o(1)}
 \label{eq:biclique-lower}
\end{equation}
queries. The same lower bound holds when the input
is promised to be bipartite with a known bipartition.
\end{theorem}

\extensionsubsection{A bipartite graph with bounded extensions}

We sample a bipartite graph and then fix a realization known to the
algorithm. Independently retaining its edges gives the query input.
The following lemma supplies the certificate count, the extension bound,
and the probability of a positive input needed in the query argument.

\begin{lemma}[A bipartite graph with bounded extensions]
\label{lem:biclique-graph}
Fix integers $2\le h<r$. For every sufficiently large $n$, there is a
bipartite graph $G$ with $n$ vertices in each part containing
$\Theta(n^{2rh/(r+h)})$ copies of $K_{r,r}$, with each matching of size
$h$ contained in at most $(\log n)^{2r+2}$ of them. Let $\cF$ be the
family of their edge sets. Independently retaining each edge of $G$
with probability $p\coloneqq|\cF|^{-1/r^2}$ leaves a copy of $K_{r,r}$
with probability $\Omega(1)$.
\end{lemma}

\begin{proof}
An extension of a fixed $K_{h,h}$ to $K_{r,r}$ requires $2(r-h)$
additional vertices and $r^2-h^2$ additional edges. The density
\begin{equation}
 \rho\coloneqq n^{-2/(r+h)}
 \label{eq:biclique-graph-parameters}
\end{equation}
makes $n^{2(r-h)}\rho^{r^2-h^2}=1$. Sample $G$ with two parts of size
$n$, taking each edge between the parts independently with probability
$\rho$.

First, let $\cF$ be the family of edge sets of copies of $K_{r,r}$ in
$G$. Each choice of $r$ vertices from each
part specifies one possible copy, so
\[
 \E|\cF|=\binom nr^2\rho^{r^2}
 =\left(\frac1{(r!)^2}+o(1)\right)n^{2rh/(r+h)}.
\]
Pairs sharing $i$ vertices on the left and $j$ on the right share $ij$
edges. There are $O(n^{4r-i-j})$ ordered pairs of this type, and both
copies occur with probability $\rho^{2r^2-ij}$. Pairs with $ij=0$ have
independent indicators and contribute zero covariance. Including
identical copies in the term $i=j=r$ gives
\[
 \frac{\operatorname{Var}(|\cF|)}{(\E|\cF|)^2}
 \le\sum_{i,j=1}^r O\left(n^{-i-j+2ij/(r+h)}\right)=o(1),
\]
since $2ij\le r(i+j)$ makes every exponent negative. Chebyshev's
inequality therefore gives
$|\cF|=(1+o(1))\E|\cF|=\Theta(n^{2rh/(r+h)})$ with probability $1-o(1)$.

Next, we bound all extension counts simultaneously. Fix sets $R$ and
$S$ of $h$ vertices in the two parts, condition on their $K_{h,h}$ being
present, and let $D$ count its extensions to $K_{r,r}$. We claim that,
for a constant $c>0$ depending only on $r,h$ and every integer $z\ge1$,
\begin{equation}
 \E[D^z\mid R\times S\subseteq E(G)]\le(cz)^{2rz}.
 \label{eq:biclique-extension-moment}
\end{equation}
Expand $D^z$ over ordered tuples of possible extensions, allowing
repetitions. Given a fixed prefix, condition on all its required edges
being present. All other edges remain independent
$\operatorname{Bernoulli}(\rho)$ variables.

If the next copy shares $i$ left and $j$ right vertices with the union
of the prefix and the fixed $K_{h,h}$, then $h\le i,j\le r$. This union
has at most $2rz$ vertices, so there are at most
$(cz)^{2r}n^{2r-i-j}$ choices for the next copy. At least $r^2-ij$ of
its edges have not yet been required, giving conditional probability at
most $\rho^{r^2-ij}$. The exponent of $n$ in their product satisfies
\[
 2r-i-j-\frac{2(r^2-ij)}{r+h}
 =-\frac{(i-h)(r-j)+(j-h)(r-i)}{r+h}\le0.
\]
Summing over the finitely many overlap types, and increasing $c$ if
necessary, bounds the conditional expected number of choices for the
next copy by $(cz)^{2r}$. Sum over the last copy for each fixed prefix
and repeat at the preceding positions to obtain
\[
 \E[D^z\mid R\times S\subseteq E(G)]
 \le(cz)^{2r}\E[D^{z-1}\mid R\times S\subseteq E(G)]
 \le(cz)^{2rz}.
\]
This proves \cref{eq:biclique-extension-moment}.

Take $z\coloneqq\lceil\log n\rceil$. Markov's inequality and a union
bound over at most $n^{2h}$ choices of $R,S$ bound the probability that
some present $K_{h,h}$ has more than $(\log n)^{2r+2}$ extensions by
\[
 n^{2h}\frac{(cz)^{2rz}}{(\log n)^{(2r+2)z}}
 =\exp\left(2h\log n-2z\log\log n+O(z)\right)=o(1).
\]
Every copy of $K_{r,r}$ containing an $h$-matching contains the
$K_{h,h}$ on its endpoints, so the same bound applies to each matching.
Let $\mathcal E$ be the event that this extension bound and the preceding
certificate count both hold. We have $\Prob_G(\mathcal E)=1-o(1)$.

It remains to ensure that a copy survives with constant probability.
For $G\in\mathcal E$, the choice $p=|\cF|^{-1/r^2}$ satisfies
\[
 p=\Theta(n^{-2h/[r(r+h)]}),\qquad \rho p=\Theta(n^{-2/r}).
\]
Let $W$ be the overlap sum in \Cref{lem:second-moment} for $\cF$ at
density $p$. The pair count above and this uniform bound on $p$ give
\[
 \E_G[W\mid\mathcal E]
 \le\frac{\sum_{\substack{1\le i,j\le r\\(i,j)\ne(r,r)}}
       O\left(n^{-i-j+2ij/r}\right)}{\Prob_G(\mathcal E)}=o(1),
\]
since $\Prob_G(\mathcal E)=1-o(1)$ and
$i(r-j)+j(r-i)>0$ for every pair in the sum. By averaging, some
$G\in\mathcal E$ satisfies $W=o(1)$. For this graph,
\Cref{lem:second-moment} gives probability at least $1/(2+o(1))$ of
retaining a copy of $K_{r,r}$, proving the lemma.
\end{proof}

\extensionsubsection{Proof of the lower bound}

We use the following consequence of K\"onig's theorem to bound the
number of edges below a threshold on matching number.

\begin{lemma}[Bipartite edge bound from matching number]
\label{lem:bipartite-extremal}
For integers $1\le h\le r$,
\begin{equation}
  m(K_{r,r},h)=r(h-1).
  \label{eq:bipartite-d}
\end{equation}
\end{lemma}

\begin{proof}
Consider an edge set $A\subseteq E(K_{r,r})$ with $\nu(A)<h$.
By K\"onig's theorem, $A$ has a vertex cover of size at most $h-1$.
Each vertex meets at most $r$ edges, so $|A|\le r(h-1)$.
Conversely, all edges incident to $h-1$ fixed vertices in one part
give a set of $r(h-1)$ edges with matching number $h-1$.
\end{proof}

\begin{proof}[Proof of \Cref{thm:biclique}]
Apply \Cref{lem:biclique-graph} with $\lfloor n/2\rfloor$ vertices in
each part, adding an isolated vertex if $n$ is odd. Let $G$ be the
resulting graph and $\cF$ its certificate family.
Restrict the input to subgraphs of $G$, as in
\Cref{subsec:certificate-preliminaries}. The coordinates are
$I\coloneqq E(G)$, each certificate has $r^2$ edges, and
$p\coloneqq|\cF|^{-1/r^2}$. Every input is bipartite with the known
bipartition of $G$. Consider a quantum query algorithm that makes at
most $t$ queries and outputs a copy of $K_{r,r}$ with probability at
least $2/3$ on every positive input. By \Cref{lem:biclique-graph}, its
average success probability under $\mu_p$ satisfies $\sigma=\Omega(1)$.

We first bound the size of a set below threshold. For each $C\in\cF$,
let
\[
 \mathcal L_C\coloneqq\{A\subseteq C:\nu(A)<h\}.
\]
Thus $A$ is below threshold when it contains no matching of size $h$.
This family is downward closed, and \Cref{lem:bipartite-extremal} gives
$|A|\le r(h-1)$ for every $A\in\mathcal L_C$.

We next count the certificates crossing the threshold when an edge $e$
is added to a database $S\subseteq I\setminus\{e\}$. As in the proof
of \Cref{prop:matching-clique}, every $C\in\cF(e,S)$ contains an
$h$-matching $J\cup\{e\}$ with $J\in\binom S{h-1}$. By
\Cref{lem:biclique-graph}, each such matching lies in at most
$\delta\coloneqq(\log n)^{2r+2}$ certificates. Hence
\[
 |\cF(e,S)|\le\delta\binom{|S|}{h-1}.
\]
The parameters in \Cref{cor:certificate-search} are therefore the
certificate size $r^2$, the low size bound $d=r(h-1)$, and $i=h-1$.
Since $p^{r^2-r(h-1)}=o(1)$ and $\sigma=\Omega(1)$, the
corollary's size condition holds for sufficiently large $n$. With
$r,h$ fixed, it gives
\[
 t=\Omega\left(\left(\frac1{\delta p^{r^2-r(h-1)+(h-1)/2}}\right)^{1/(h+1)}\right)
 =\Omega\left(\delta^{-1/(h+1)}
   |\cF|^{\frac1{h+1}\left(1-\frac{h-1}{r}+\frac{h-1}{2r^2}\right)}\right).
\]
Substituting $|\cF|=\Theta(n^{2rh/(r+h)})$ and
$\delta=(\log n)^{2r+2}$ proves the search bound in
\cref{eq:biclique-lower}, even for bipartite inputs with a known
bipartition. The equivalence of decision and search for fixed connected subgraphs also holds in this setting. It therefore gives the same lower bound for decision.
\end{proof}

We now optimize the threshold. For each integer $r\ge3$, define
\begin{equation}
 \beta_r\coloneqq\max_{2\le h<r}
   \frac{2hr}{(h+1)(r+h)}
   \left(1-\frac{h-1}{r}+\frac{h-1}{2r^2}\right),
 \label{eq:biclique-exponent}
\end{equation}
where the maximum is over integer thresholds. By \Cref{thm:biclique},
$K_{r,r}$ detection and search require $n^{\beta_r-o(1)}$ queries.

\begin{corollary}[Choosing the matching threshold]
\label{cor:biclique-limits}
The sequence $(\beta_r)_{r\ge3}$ is strictly increasing, with
$\beta_{10}=181/180$.
Moreover, as $r\to\infty$,
\begin{equation}
 \beta_r\ge2-\frac{4\sqrt2}{\sqrt r}+O(r^{-1}).
 \label{eq:biclique-near-quadratic}
\end{equation}
\end{corollary}

\begin{proof}
Fix $h\ge2$. For $r>h$, the factors $r/(r+h)$ and
$1-(h-1)/r+(h-1)/(2r^2)$ in \cref{eq:biclique-exponent} are positive
and strictly increasing with $r$. A threshold attaining the maximum at $r$ remains
admissible at $r+1$, where its exponent is strictly larger. Hence
$\beta_{r+1}>\beta_r$.

At $r=10$, the maximum in \cref{eq:biclique-exponent} is attained at
$h=2$, giving $\beta_{10}=181/180$.

To choose $h$ for large $r$, we express the difference between two and our exponent as a sum of two nonnegative terms:
\[
 2-\frac{2hr}{(h+1)(r+h)}
 \left(1-\frac{h-1}{r}+\frac{h-1}{2r^2}\right)
 =\frac2{h+1}+\frac{2h}{h+1}
   \frac{2h-1-(h-1)/(2r)}{r+h}.
\]
Since $2\le h<r$, the second term satisfies
\[
 \frac{2h}{h+1}\frac{2h-1-(h-1)/(2r)}{r+h}
 \ge\frac{h}{r+h}\ge\frac{h}{2r}.
\]
An exponent approaching two therefore requires $h\to\infty$ because of the first term and
$h=o(r)$ because of the second term; bounded $h$ or $h=\Omega(r)$ leaves a constant loss.
In this regime, the two terms in the loss are asymptotic to $2/h$ and
$4h/r$, respectively. The leading loss
$2/h+4h/r$ is minimized at $h=\sqrt{r/2}$.
Choose a nearest integer, so that $h=\sqrt{r/2}+O(1)$; this satisfies
$2\le h<r$ for sufficiently large $r$. For this choice,
\[
 \frac{2h}{h+1}\frac{r}{r+h}
 \left(1-\frac{h-1}{r}+\frac{h-1}{2r^2}\right)
 =2-\frac2h-\frac{4h}{r}+O(r^{-1}).
\]
This proves \cref{eq:biclique-near-quadratic}.
\end{proof}

\Cref{thm:biclique,cor:biclique-limits} together prove
\Cref{thm:biclique-intro}.

\endgroup

\section*{Acknowledgments}
\addcontentsline{toc}{section}{Acknowledgments}
\ifanonymous\else
We thank Aleksandrs Belovs for suggesting that his framework based on
a hidden type could be applied to lower bounds for clique detection. ASG 
received support from the National Science Foundation (grant 26-17356) 
and the Department of Energy (grant DE-SC0020264 and the Office of 
Science, Office of Advanced Scientific Computing Research, Accelerated Research in Quantum Computing program). XZ is
supported by NSERC Grant RGPIN-2024-06493.
\fi

\paragraph{Use of generative AI.}
OpenAI GPT 5.5, 5.6 Sol, 6 Astra, and Anthropic Opus 5.5 were used for literature searches, manuscript drafting and revision, proof development and checking, and LaTeX preparation. The  authors retain full responsibility for reviewing the generated material and for all mathematical claims and final text.

In particular, the authors used GPT 5.5 and GPT 5.6 Sol to fill in technical gaps in the authors' blueprint for \Cref{thm:certificate-search} and also to find specific graph constructions with certain properties (such as \cite{GowersJanzer21}) that were used to construct hard instances for our lower bounds. 

\phantomsection
\addcontentsline{toc}{section}{References}
\bibliographystyle{alphaurl}
\bibliography{references}

\newcommand{\etalchar}[1]{$^{#1}$}
\begin{thebibliography}{GWWZ26}

\bibitem[Amb07]{Ambainis07}
Andris Ambainis.
\newblock Quantum walk algorithm for element distinctness.
\newblock {\em SIAM Journal on Computing}, 37(1):210--239, 2007.
\newblock \href {https://doi.org/10.1137/S0097539705447311} {\path{doi:10.1137/S0097539705447311}}.

\bibitem[Bel12]{Belovs12}
Aleksandrs Belovs.
\newblock Span programs for functions with constant-sized {$1$}-certificates.
\newblock In {\em Proceedings of the 44th Annual ACM Symposium on Theory of Computing}, pages 77--84, 2012.
\newblock \href {https://doi.org/10.1145/2213977.2213985} {\path{doi:10.1145/2213977.2213985}}.

\bibitem[Bel26]{Belovs26}
Aleksandrs Belovs.
\newblock Tight quantum lower bound for {$k$}-distinctness.
\newblock In {\em 67th Annual {IEEE} Symposium on Foundations of Computer Science ({FOCS} 2026)}, 2026.
\newblock To appear.
\newblock \href {https://arxiv.org/abs/2604.05133} {\path{arXiv:2604.05133}}.

\bibitem[BI16]{BalodisIraids16}
Kaspars Balodis and J\=anis Iraids.
\newblock Quantum lower bound for graph collision implies lower bound for triangle detection.
\newblock {\em Baltic Journal of Modern Computing}, 4(4):731--735, 2016.
\newblock \href {https://doi.org/10.22364/bjmc.2016.4.4.10} {\path{doi:10.22364/bjmc.2016.4.4.10}}.

\bibitem[BR12]{BelovsReichardt12}
Aleksandrs Belovs and Ben~W. Reichardt.
\newblock Span programs and quantum algorithms for {$st$}-connectivity and claw detection.
\newblock In {\em 20th Annual European Symposium on Algorithms}, volume 7501 of {\em Lecture Notes in Computer Science}, pages 193--204, 2012.
\newblock \href {https://doi.org/10.1007/978-3-642-33090-2_18} {\path{doi:10.1007/978-3-642-33090-2_18}}.

\bibitem[BR14]{BelovsRosmanis14}
Aleksandrs Belovs and Ansis Rosmanis.
\newblock On the power of non-adaptive learning graphs.
\newblock {\em Computational Complexity}, 23(2):323--354, 2014.
\newblock \href {https://doi.org/10.1007/s00037-014-0084-1} {\path{doi:10.1007/s00037-014-0084-1}}.

\bibitem[B{\v S}13]{BelovsSpalek13}
Aleksandrs Belovs and Robert {{\v S}palek}.
\newblock Adversary lower bound for the {$k$}-sum problem.
\newblock In {\em Proceedings of the 4th Innovations in Theoretical Computer Science Conference}, pages 323--328, 2013.
\newblock \href {https://doi.org/10.1145/2422436.2422474} {\path{doi:10.1145/2422436.2422474}}.

\bibitem[CE05]{ChildsEisenberg05}
Andrew~M. Childs and Jason~M. Eisenberg.
\newblock Quantum algorithms for subset finding.
\newblock {\em Quantum Information and Computation}, 5(7):593--604, 2005.
\newblock \href {https://doi.org/10.26421/QIC5.7-7} {\path{doi:10.26421/QIC5.7-7}}.

\bibitem[CFHL21]{ChungEtAl21}
Kai-Min Chung, Serge Fehr, Yu-Hsuan Huang, and Tai-Ning Liao.
\newblock On the compressed-oracle technique, and post-quantum security of proofs of sequential work.
\newblock In {\em Advances in Cryptology---{EUROCRYPT} 2021}, volume 12697 of {\em Lecture Notes in Computer Science}, pages 598--629, 2021.
\newblock \href {https://doi.org/10.1007/978-3-030-77886-6_21} {\path{doi:10.1007/978-3-030-77886-6_21}}.

\bibitem[CGK{\etalchar{+}}23]{CojocaruEtAl23}
Alexandru Cojocaru, Juan Garay, Aggelos Kiayias, Fang Song, and Petros Wallden.
\newblock Quantum multi-solution {Bernoulli} search with applications to {Bitcoin}'s post-quantum security.
\newblock {\em Quantum}, 7:944, 2023.
\newblock \href {https://doi.org/10.22331/q-2023-03-09-944} {\path{doi:10.22331/q-2023-03-09-944}}.

\bibitem[CGP26]{CornelissenGilaniPatro26}
Arjan Cornelissen, Amin~Shiraz Gilani, and Subhasree Patro.
\newblock Quantum algorithms for path and cycle containment problems.
\newblock In {\em Approximation, Randomization, and Combinatorial Optimization. Algorithms and Techniques ({APPROX/RANDOM} 2026)}, volume 392 of {\em Leibniz International Proceedings in Informatics}, pages 72:1--72:23, 2026.
\newblock \href {https://doi.org/10.4230/LIPIcs.APPROX/RANDOM.2026.72} {\path{doi:10.4230/LIPIcs.APPROX/RANDOM.2026.72}}.

\bibitem[CK11]{ChildsKothari11}
Andrew~M. Childs and Robin Kothari.
\newblock Quantum query complexity of minor-closed graph properties.
\newblock In {\em 28th International Symposium on Theoretical Aspects of Computer Science}, volume~9 of {\em Leibniz International Proceedings in Informatics}, pages 661--672, 2011.
\newblock \href {https://doi.org/10.4230/LIPIcs.STACS.2011.661} {\path{doi:10.4230/LIPIcs.STACS.2011.661}}.

\bibitem[CLM20]{CaretteEtAl20}
Titouan Carette, Mathieu Lauri\`ere, and Fr\'ed\'eric Magniez.
\newblock Extended learning graphs for triangle finding.
\newblock {\em Algorithmica}, 82(4):980--1005, 2020.
\newblock \href {https://doi.org/10.1007/s00453-019-00627-z} {\path{doi:10.1007/s00453-019-00627-z}}.

\bibitem[Dir52]{Dirac52}
G.~A. Dirac.
\newblock A property of 4-chromatic graphs and some remarks on critical graphs.
\newblock {\em Journal of the London Mathematical Society}, s1-27(1):85--92, 1952.
\newblock \href {https://doi.org/10.1112/jlms/s1-27.1.85} {\path{doi:10.1112/jlms/s1-27.1.85}}.

\bibitem[DP25]{DelcourtPostle24}
Michelle Delcourt and Luke Postle.
\newblock Reducing linear {Hadwiger}'s conjecture to coloring small graphs.
\newblock {\em Journal of the American Mathematical Society}, 38(2):481--507, 2025.
\newblock \href {https://doi.org/10.1090/jams/1047} {\path{doi:10.1090/jams/1047}}.

\bibitem[DVW22]{DalirrooyfardWilliams22}
Mina Dalirrooyfard and Virginia Vassilevska~Williams.
\newblock Induced cycles and paths are harder than you think.
\newblock In {\em 63rd Annual {IEEE} Symposium on Foundations of Computer Science}, pages 531--542, 2022.
\newblock \href {https://doi.org/10.1109/FOCS54457.2022.00057} {\path{doi:10.1109/FOCS54457.2022.00057}}.

\bibitem[EG59]{ErdosGallai59}
Paul Erd{\H{o}}s and Tibor Gallai.
\newblock On maximal paths and circuits of graphs.
\newblock {\em Acta Mathematica Academiae Scientiarum Hungaricae}, 10(3--4):337--356, 1959.
\newblock \href {https://doi.org/10.1007/BF02024498} {\path{doi:10.1007/BF02024498}}.

\bibitem[GJ21]{GowersJanzer21}
W.~T. Gowers and Barnab{\'a}s Janzer.
\newblock Generalizations of the {Ruzsa--Szemer{\'e}di} and rainbow {Tur{\'a}n} problems for cliques.
\newblock {\em Combinatorics, Probability and Computing}, 30(4):591--608, 2021.
\newblock \href {https://doi.org/10.1017/S0963548320000589} {\path{doi:10.1017/S0963548320000589}}.

\bibitem[GWWZ26]{GilaniEtAl26}
Amin~Shiraz Gilani, Daochen Wang, Pei Wu, and Xingyu Zhou.
\newblock Quantum algorithms on edge lists: Hiding, shuffling, and cycle finding.
\newblock In {\em 53rd International Colloquium on Automata, Languages, and Programming}, volume 374 of {\em Leibniz International Proceedings in Informatics}, pages 97:1--97:16, 2026.
\newblock \href {https://doi.org/10.4230/LIPIcs.ICALP.2026.97} {\path{doi:10.4230/LIPIcs.ICALP.2026.97}}.

\bibitem[HHL26]{HaoHuangLiu26}
Zihan Hao, Zikuan Huang, and Qipeng Liu.
\newblock On the need for (quantum) memory with short outputs.
\newblock In {\em Proceedings of the 58th Annual {ACM} Symposium on Theory of Computing}, pages 1901--1912, 2026.
\newblock \href {https://doi.org/10.1145/3798129.3800896} {\path{doi:10.1145/3798129.3800896}}.

\bibitem[HM23]{HamoudiMagniez23}
Yassine Hamoudi and Fr\'ed\'eric Magniez.
\newblock Quantum time--space tradeoff for finding multiple collision pairs.
\newblock {\em ACM Transactions on Computation Theory}, 15(1--2):3:1--3:22, 2023.
\newblock \href {https://doi.org/10.1145/3589986} {\path{doi:10.1145/3589986}}.

\bibitem[JKM13]{JefferyKothariMagniez13}
Stacey Jeffery, Robin Kothari, and Fr{\'e}d{\'e}ric Magniez.
\newblock Nested quantum walks with quantum data structures.
\newblock In {\em Proceedings of the 24th Annual ACM-SIAM Symposium on Discrete Algorithms}, pages 1474--1485, 2013.
\newblock \href {https://doi.org/10.1137/1.9781611973105.106} {\path{doi:10.1137/1.9781611973105.106}}.

\bibitem[JZ26]{JefferyZur26}
Stacey Jeffery and Sebastian Zur.
\newblock The compressed oracle is a worthy (multiplicative) adversary.
\newblock In {\em 53rd International Colloquium on Automata, Languages, and Programming}, volume 374 of {\em Leibniz International Proceedings in Informatics}, pages 118:1--118:23, 2026.
\newblock \href {https://doi.org/10.4230/LIPIcs.ICALP.2026.118} {\path{doi:10.4230/LIPIcs.ICALP.2026.118}}.

\bibitem[LG14]{LeGall14}
Fran{\c c}ois Le~Gall.
\newblock Improved quantum algorithm for triangle finding via combinatorial arguments.
\newblock In {\em 55th Annual {IEEE} Symposium on Foundations of Computer Science}, pages 216--225, 2014.
\newblock \href {https://doi.org/10.1109/FOCS.2014.31} {\path{doi:10.1109/FOCS.2014.31}}.

\bibitem[LMR{\etalchar{+}}11]{LMRSS11}
Troy Lee, Rajat Mittal, Ben~W. Reichardt, Robert {{\v S}palek}, and Mario Szegedy.
\newblock Quantum query complexity of state conversion.
\newblock In {\em 52nd Annual {IEEE} Symposium on Foundations of Computer Science}, pages 344--353, 2011.
\newblock \href {https://doi.org/10.1109/FOCS.2011.75} {\path{doi:10.1109/FOCS.2011.75}}.

\bibitem[LMS12]{LeeMagniezSantha12}
Troy Lee, Fr\'ed\'eric Magniez, and Miklos Santha.
\newblock Learning graph based quantum query algorithms for finding constant-size subgraphs.
\newblock {\em Chicago Journal of Theoretical Computer Science}, 2012(10):1--21, 2012.
\newblock \href {https://doi.org/10.4086/cjtcs.2012.010} {\path{doi:10.4086/cjtcs.2012.010}}.

\bibitem[LMS17]{LeeMagniezSantha13}
Troy Lee, Fr{\'e}d{\'e}ric Magniez, and Miklos Santha.
\newblock Improved quantum query algorithms for triangle detection and associativity testing.
\newblock {\em Algorithmica}, 77(2):459--486, 2017.
\newblock \href {https://doi.org/10.1007/s00453-015-0084-9} {\path{doi:10.1007/s00453-015-0084-9}}.

\bibitem[MNRS11]{MNRS11}
Fr{\'e}d{\'e}ric Magniez, Ashwin Nayak, J{\'e}r{\'e}mie Roland, and Miklos Santha.
\newblock Search via quantum walk.
\newblock {\em SIAM Journal on Computing}, 40(1):142--164, 2011.
\newblock \href {https://doi.org/10.1137/090745854} {\path{doi:10.1137/090745854}}.

\bibitem[MSS07]{MagniezSanthaSzegedy07}
Fr{\'e}d{\'e}ric Magniez, Miklos Santha, and Mario Szegedy.
\newblock Quantum algorithms for the triangle problem.
\newblock {\em SIAM Journal on Computing}, 37(2):413--424, 2007.
\newblock \href {https://doi.org/10.1137/050643684} {\path{doi:10.1137/050643684}}.

\bibitem[RST93]{RobertsonSeymourThomas93}
Neil Robertson, Paul Seymour, and Robin Thomas.
\newblock Hadwiger's conjecture for {$K_6$}-free graphs.
\newblock {\em Combinatorica}, 13(3):279--361, 1993.
\newblock \href {https://doi.org/10.1007/BF01202354} {\path{doi:10.1007/BF01202354}}.

\bibitem[{\v{S}}S05]{SpalekSzegedy06}
Robert {\v{S}}palek and Mario Szegedy.
\newblock All quantum adversary methods are equivalent.
\newblock In {\em Proceedings of the 32nd International Colloquium on Automata, Languages and Programming (ICALP 2005)}, volume 3580, pages 1299--1311, 2005.
\newblock \href {https://doi.org/10.1007/11523468_105} {\path{doi:10.1007/11523468_105}}.

\bibitem[TM24]{TeraoMori24}
Tatsuya Terao and Ryuhei Mori.
\newblock Parameterized quantum query algorithms for graph problems.
\newblock In {\em 32nd Annual European Symposium on Algorithms}, volume 308 of {\em Leibniz International Proceedings in Informatics}, pages 99:1--99:16, 2024.
\newblock \href {https://doi.org/10.4230/LIPIcs.ESA.2024.99} {\path{doi:10.4230/LIPIcs.ESA.2024.99}}.

\bibitem[Zha04]{Zhang04}
Shengyu Zhang.
\newblock On the power of {A}mbainis's lower bounds.
\newblock In {\em Proceedings of the 31st International Colloquium on Automata, Languages and Programming (ICALP 2004)}, volume 3142, pages 1238--1250, 2004.
\newblock \href {https://doi.org/10.1007/978-3-540-27836-8_102} {\path{doi:10.1007/978-3-540-27836-8_102}}.

\bibitem[Zha19]{Zhandry19}
Mark Zhandry.
\newblock How to record quantum queries, and applications to quantum indifferentiability.
\newblock In {\em Advances in Cryptology---{CRYPTO} 2019}, volume 11693 of {\em Lecture Notes in Computer Science}, pages 239--268, 2019.
\newblock \href {https://doi.org/10.1007/978-3-030-26951-7_9} {\path{doi:10.1007/978-3-030-26951-7_9}}.

\bibitem[Zhu12]{Zhu12}
Yechao Zhu.
\newblock Quantum query complexity of constant-sized subgraph containment.
\newblock {\em International Journal of Quantum Information}, 10(3):1250019, 2012.
\newblock \href {https://doi.org/10.1142/S0219749912500190} {\path{doi:10.1142/S0219749912500190}}.

\end{thebibliography}

\end{document}